\documentclass[12pt]{amsart}

\usepackage{geometry}
\usepackage{fullpage}
\usepackage[utf8]{inputenc}
\usepackage{comment}
\usepackage{amsmath, amssymb}
\usepackage{mathrsfs, enumerate, csquotes, color,float}

\usepackage[hidelinks]{hyperref}
\usepackage{tikz,pgfplots}
\pgfplotsset{compat=1.13}
\usepackage[outline]{contour}
\usetikzlibrary{arrows.meta}
\usetikzlibrary{backgrounds}
\usetikzlibrary{calc}

\renewcommand{\vec}[1]{\mathbf{#1}}
\newcommand{\R}{\mathbb{R}}
\DeclareMathOperator{\supp}{supp}
\DeclareMathOperator{\dist}{dist}
\renewcommand{\Re}{\mathrm{Re}\,}
\renewcommand{\Im}{\mathrm{Im}\,}
\renewcommand{\leq}{\leqslant}	
\renewcommand{\geq}{\geqslant}
\newcommand{\cH}{\mathcal{H}}
\newcommand{\cHd}{\mathcal{H}^2}

\newtheorem{theorem}{Theorem}[section]

\newtheorem{lemma}[theorem]{Lemma}
\newtheorem{corollary}[theorem]{Corollary}
\newtheorem{proposition}[theorem]{Proposition}

\theoremstyle{definition}
\newtheorem{definition}[theorem]{Definition}
\newtheorem{remark}[theorem]{Remark}
\newtheorem{assumption}[theorem]{Assumption}

\makeatletter

\@addtoreset{equation}{section}
\makeatother

\DeclareMathOperator{\Id}{Id}

\title[Localization of the eigenfunctions of a Bloch-Torrey operator]{Localization of the eigenfunctions of a Bloch-Torrey operator on the half-plane}
\author[M. Averseng]{M. Averseng}
\address[M. Averseng]{Univ Angers, CNRS, LAREMA, F-49000 Angers, France}
\email{martin.averseng@univ-angers.fr}

\author[N. Frantz]{N. Frantz}
\address[N. Frantz]{Univ Angers, CNRS, LAREMA, F-49000 Angers, France}
\email{nicolas.frantz@univ-angers.fr}

\author[F. H\'erau]{F. H\'erau}
\address[F. H\'erau]{LMJL - UMR6629, Nantes Université, CNRS, 2 rue de la Houssini\`ere, BP 92208, F-44322 Nantes cedex 3, France}
\email{herau@univ-nantes.fr}

\author[N. Raymond]{N. Raymond}
\address[N. Raymond]{Univ Angers, CNRS, LAREMA, Institut Universitaire de France, F-49000 Angers, France}
\email{nicolas.raymond@univ-angers.fr}
\date{}

\begin{document}

\begin{abstract}
We consider a non-self adjoint operator of the form $-h^2 \Delta + i(V(x) + \alpha(x)y)$ on the upper half plane $y > 0$ with Dirichlet boundary conditions on $\{y = 0\}$ with $V \geq 0$, $V$ admitting a non-degenerate minimum at $x = 0$ and $\alpha'(0) = 0$. We study its eigenfunctions associated to the smallest eigenvalues in magnitude in the semiclassical limit $h \to 0$. Elementary variational estimates show that these eigenfunctions are localized near the point $(0,0)$ at the scales $O(h^{1/3})$ in $x$ and $O(h^{2/3})$ in $y$. In this paper, we show that the $O(h^{1/3})$ localization in $x$ is not optimal; more precisely, we establish that the eigenfunctions are concentrated in a neighborhood of size $O(h^{1/2})$ of the axis $\{x = 0\}$, and this scale is shown to be sharp.
The proof relies on the symbolic calculus of operator-valued  pseudodifferential operators.
\end{abstract}

\maketitle

\section{Motivations and questions}

\subsection{Main motivation}
The main motivation for this paper is the study of the ``low-energy eigenfunctions'' (i.e., associated to smallest eigenvalues in magnitude) of the {\em Bloch-Torrey operator}
$$\mathscr{B}_h := -h^2 \Delta + i x_1$$ 
defined on an appropriate subspace of $L^2(\Omega)$, where $\Omega \subset \R^2$ is a smooth, bounded domain,  with suitable boundary conditions, in the semiclassical limit $h \to 0$. More precisely, we are interested in analyzing the localization phenomenon that takes place at the point of $\partial \Omega$ where $x_1$ is minimal; this phenomenon is illustrated by the numerical experiments displayed in Figure \ref{fig:numerics} in the case of Dirichlet boundary conditions. 
\begin{figure}[htbp]
\includegraphics[width=0.20\textwidth]{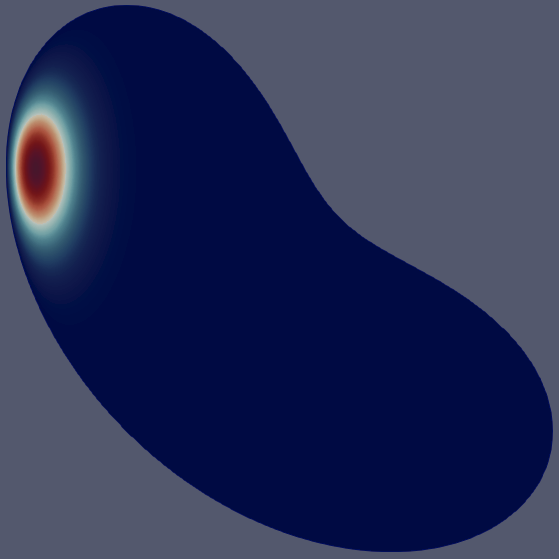}
\includegraphics[width=0.20\textwidth]{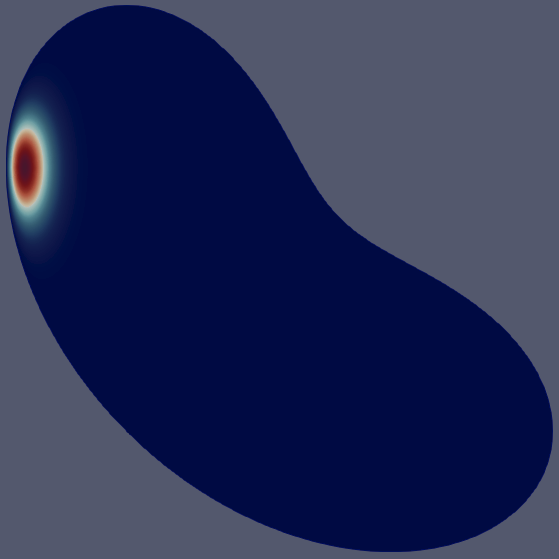}
\includegraphics[width=0.20\textwidth]{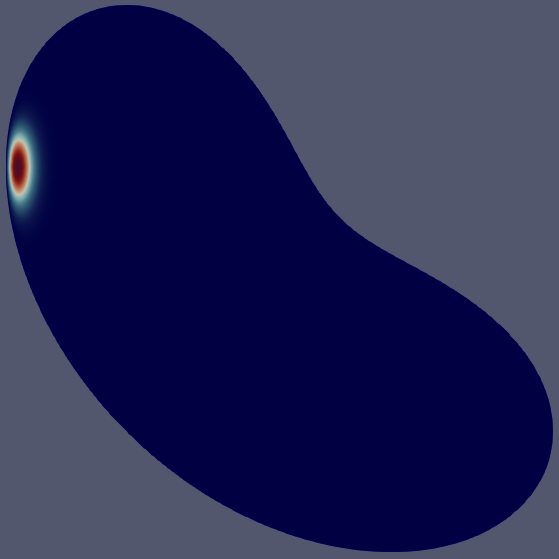}
\includegraphics[width=0.20\textwidth]{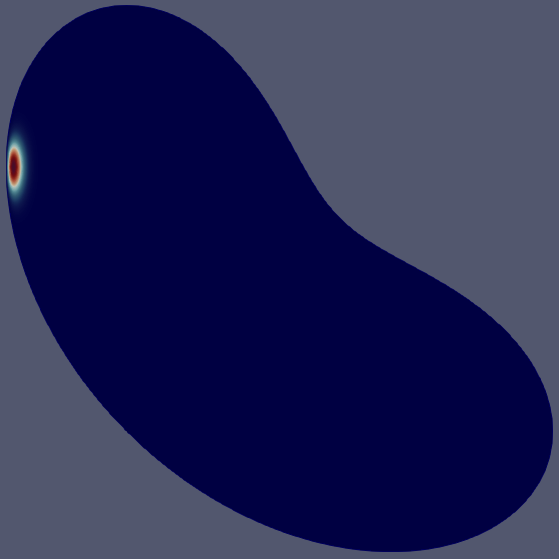}
\caption{Plot of the magnitude of the $L^2$ normalized first eigenfunction of the Bloch-Torrey operator $-(2^{-n})^2\Delta + ix_1$ with Dirichlet boundary conditions on the two-dimensional domain $\Omega$ bounded by the curve parametrized by $r(\theta) = \cos(\theta)^3 + \sin(\theta)^3$ in polar coordinates, for $n = 7,8,9,10$ (from top-left to bottom-right). The red intensity is proportional to the magnitude of the plotted eigenfunction. Numerical computations were performed using the finite element method with a refined mesh near the point of the boundary where the eigenfunction is localized. The $L^2$ norm of the eigenfunction is indeed small outside of a box of size $O(h^{1/2})$ along the $y$-axis direction and $O(h^{2/3})$ along the $x$-axis, where $h = 2^{-n}$. The source code for these computations is available at \cite{code}.}
\label{fig:numerics}
\end{figure}
Our goal is to tackle the following question:
\begin{center}
{\bf Q1:} {\em How are the low-energy eigenfunctions of the Bloch-Torrey operator localized?}
\end{center}
In this work, we start addressing this question by considering a model operator on a half-plane, obtained by local flattening near the point of minimal $x_1$-coordinate, and we hope to address {\bf Q1} in its full generality in a subsequent work.

\subsection{Physical background}

In their simplest form -- in dimensionless units, without relaxation and with isotropic diffusion -- the {\em Bloch-Torrey equations} read 
$$\frac{dM}{dt}(\vec x,t) = -i gz M + \Delta M$$
where $M$ is the transverse magnetization of spin-bearing particles in a domain $\Omega$, subject to a magnetic field with constant gradient $g$ in the $z$ direction. 

These equations were derived in 1956 by Torrey \cite{torrey1956bloch} from the Bloch equations \cite{bloch1946nuclear} to take into account diffusion effects due to inhomogenous magnetic fields. They are used to describe the magnetization diffusion of nuclei in a confined domain, and are the main model for the imaging technique known as diffusion MRI, or dMRI (a variant of MRI, since the latter makes uses of constant magnetic fields), which is used for medical imaging, especially applied to brain cells (see, e.g., \cite{lebihan2003looking}). The spectrum of the operator
$$-\Delta + i g z\,,$$
(usually considered with Neumann or Robin conditions) describes the relationship between a measurable signal, obtained by spatially integrating the transverse magnetization data over a region called ``voxel'', and the microstructure of the imaged tissue see e.g. \cite{grebenkov2024spectral}. The case of a large gradient $g \gg 1$ (which is equivalent to the semiclassical regime with $h \to g^{-1/2}$), or ``localization regime'', was considered in a seminal paper \cite{stoller1991transverse}, see also \cite{grebenkov2014exploring,grebenkov2018diffusion,moutal2019localization}. Localization effects (such as the ones visible in Figure \ref{fig:numerics}) are connected to the phenomenon of ``diffusive edge enhancement'', in which boundaries of confining cells appear brighter on reconstructed images, see \cite[Section V]{deswiet1994decay}. As of yet, this localization regime is ``yet poorly understood and exploiting its potential advantages is still challenging in experiments'' \cite[Section 5]{moutal2019localization}, see also \cite{grebenkov2018diffusion} for related discussions.

We also note that the Bloch-Torrey operator appears in other physical applications, such as superconductivity. For instance, the linearized time-dependent Ginzburg-Landau equations near the ``normal state'', under the assumption of vanishing magnetic field, lead to a model close to the Bloch-Torrey equations, see, e.g. \cite[Section 2.2]{almog2008stability}.

\subsection{Known results}

The mathematical investigation of the spectrum of the Bloch-Torrey operator was started by Almog in \cite{almog2008stability} in the context of superconductivity. 
This work was followed by several others, (e.g. \cite{henry2014semiclassical,almog2016spectral,almog2018spectral,grebenkov2018spectral}) which, together, show in various geometric settings and choices of boundary conditions, that the ``left-margin'' of the spectrum $\sigma(\mathscr{B}_h)$ satisfies 
$$ \lim_{h \to 0} \Re(\sigma(\mathscr{B}_h)) = \frac{|z_1|}{2}h^{2/3}$$
where $z_1 \approx -2.33811$ is the rightmost real zero of the Airy function. 
In particular, an upper bound on this quantity is obtained in \cite{almog2016spectral} by showing that (under adequate non-degeneracy assumptions on the boundary) in the case of Dirichlet boundary conditions, there exists an eigenvalue $\lambda(h)$ of $\mathscr{B}_h$ satisfying
$$\lambda(h) = i \left(\inf_{\Omega}x_1\right) + e^{i\pi/3}|z_1| h^{2/3} + Kh + o(h),$$
where $K$ is a constant related to the curvature of the boundary of $\Omega$ at the point of minimal abscissa (this is a particular case of \cite[Theorem 1.1]{almog2016spectral} choosing $V = x_1$). The work \cite{grebenkov2018spectral} furthermore constructs ``quasimodes'', with associated ``quasi-eigenvalues'' close to $\lambda(h)$ -- i.e., approximate solutions $(\widetilde{\lambda}(h),\widetilde{u}(h))$ of the eigenvalue problem 
$$\big(\mathscr{B}_h - \widetilde{\lambda}(h)\big) \widetilde{u}(h) = 0\,,$$
with $\widetilde{\lambda}(h) \approx \lambda(h)$. Their ansatz $\widetilde{u}(h)$ possess a localization property which is similar to the one visible in Figure \ref{fig:numerics}. However, the fact that these quasimodes are actually close to true eigenfunctions is not known.

Recently, analytic dilation techniques were used in \cite{herau2024semiclassical} to obtain approximations of the low-energy eigenvalues of the more general operator $-h^2 \Delta + e^{i\alpha} x_1$, $\alpha \in [0,\frac{3\pi}{5})$. These approximations agree with \cite[Equation (23)]{deswiet1994decay} when $\alpha=\frac{\pi}{2}$. The eigenfunctions of the resulting operator after analytic dilation are localized in a way that is completely understood. Unfortunately, dilating back to the original coordinates, the localization information is lost: only the spectrum is preserved. 

The above results and their proofs suggest that the natural localization for the eigenfunctions is $O(h^{2/3})$ in the normal direction and $O(h^{1/2})$ in the tangential direction. We also highlight that these localization results are known for a self-adjoint counterpart of the Bloch-Torrey operator (i.e., replacing $ix_1$ by $x_1$ in the definition of $\mathscr{B}_h$), see \cite{cornean2022two}; however, as far as we are aware, they have not been proven for the complex version. Moreover, while the $O(h^{2/3})$ localization in the normal direction can be obtained via fairly standard arguments (see Remark \ref{rem:comments_main} (\ref{remitem:pollution})), obtaining the $O(h^{1/2})$ scale in the tangential variable appears to require new techniques and seems to be the main difficulty for answering {\bf Q1}. In this paper, we prove this $O(h^{1/2})$ localization for a model operator on the half-plane.

\subsection{Model operator}

As suggested by the above discussion, we should focus on the point of $\partial \Omega$ where $x_1$ is minimal, which we can assume to be located at $(0,0)$. Introducing a normal parametrization $\gamma : [0,1] \to \R^2$ of $\partial \Omega$, we can consider the system of coordinates
$$x(s,t) = \vec \gamma(s) + t N(s)\,, \quad (s,t) \in (-s_0,s_0) \times (0,t_0)$$
in a tubular neighborhood of $(0,0)$, where $N(s)$ is the unit normal vector at $\gamma(s)$ pointing inside $\Omega$ (see Figure \ref{fig:param}).

\begin{figure}[htbp]
\begin{tikzpicture}[scale=0.7]
\draw  (-3,0) -- (3,0);
\draw  (0,-3) -- (0,3);
\draw   plot[smooth,domain=-3:3] (0.3*\x*\x, {\x});
\node at (0.3*2*2+0.5,2-0.5) {$\Omega$};
\node at (0.3*9+0.5,3+0.1) {$\partial\Omega$};
\draw[-Latex] (0.3*4,-2)--(0.3*4+0.4*1,-2+0.4*1.2) node[below right=-1mm]{$N(s)$};
\draw[dashed,blue] (0.3*4,-2)--(0.3*4+1.2*1,-2+1.2*1.2) node[right]{$x(s,t)$};
\draw [decorate,decoration={brace,amplitude=10pt},blue] (0.3*4,-2)--(0.3*4+1.2*1,-2+1.2*1.2) node[midway,xshift=-0.9em,yshift=1.08em] {$t$};
\node[below left] at (0.3*4+0.15,-2) {$\gamma(s)$};
\end{tikzpicture}
\caption{Local parametrization of the boundary $\partial \Omega$ near the point of minimal abscissa.}
\label{fig:param}
\end{figure}
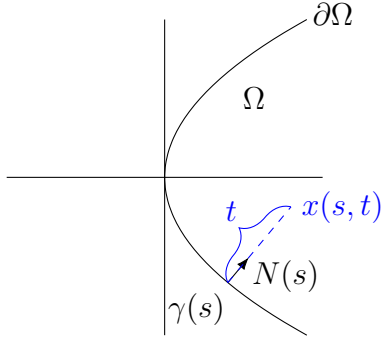
Discarding some curvature terms, we are led to considering the model operator 
$$\mathscr{T}_h := -h^2 (\partial_s^2 + \partial_t^2) + i\alpha(s) t + i V(s)$$
on the open set $(s,t) \in \R \times \R_+$, where $\alpha(s)$ and $V(s)$ coincide with $N(s) \cdot e_1$ and $V(s) := \gamma(s) \cdot e_1$ respectively in a neighborhood of $s = 0$. In particular, assuming that the curvature does not vanish at $s = 0$, the functions $\alpha$ and $V$ can be chosen to satisfy the following properties: 
\begin{enumerate}
\item[(i)] $\alpha(0) = 1$, $\alpha'(0) = 0$ and $\alpha_0 \leq \alpha(s) \leq 1$ for some $\alpha_0 > 0$.
\item[(ii)] $V(0) = V'(0) = 0$ and $V''(0) > 0$. 
\end{enumerate}
Through the change of variables $t \to h^{2/3} y$ and $s \to x$, this operator is unitarily equivalent to 
$$\mathscr{L}_h = h^{2/3} \big(D_y^2 + i\alpha(x)y\big) + h^2 D_x^2 + iV(x),$$
where here and in what follows we denote $D = -i\partial$. Numerical experiments displayed in Figures \ref{fig:numericsLh} and \ref{fig:Airy} suggest that the low-lying eigenfunctions of $\mathscr{L}_h$ concentrate near the half-line $\{x = 0\}$ as $h \to 0$, and are approximately of the form $\psi_h(x,y) \approx f_h(x) u_{\rm Ai}(y)$ where $u_{\rm Ai}(y)$ is a ``ground-state'' of the one-dimensional, complex Airy operator 
$D_y^2 + i\alpha(0)y$
with Dirichlet conditions on the half-line (one such eigenfunction is explicitly given by $u_{\rm Ai}(y) = \textup{Ai}(e^{i \frac\pi6} \alpha(0)^{1/3}y + z_1)$), and where $f_h$ is localized at the scale $O(h^{1/2})$.
The goal of this paper is to describe this localization in the $x$-variable i.e., answer the question
\begin{center}
{\bf Q2:} {\em At what scale are the low-lying eigenfunctions of $\mathscr{L}_h$ localized in the $x$-variable ?}
\end{center}
We expect that the answer to {\bf Q2} is the key difficulty in answering {\bf Q1}. More precisely, the ellipticity of $\mathscr{L}_h$ (see Proposition \ref{prop:main_ellip} below) should be the main ingredient for obtaining the tangential localization in the original problem. 

\begin{figure}[htbp]
\includegraphics[width=0.20\textwidth]{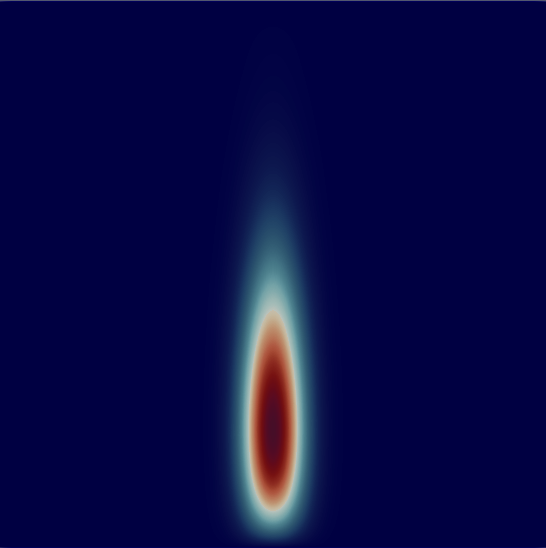}
\includegraphics[width=0.20\textwidth]{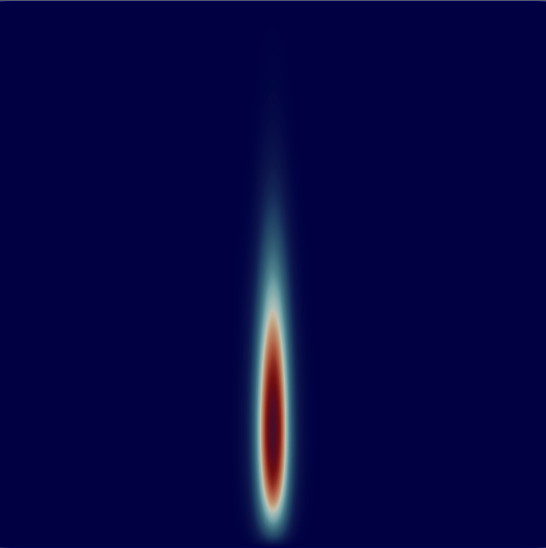}
\includegraphics[width=0.20\textwidth]{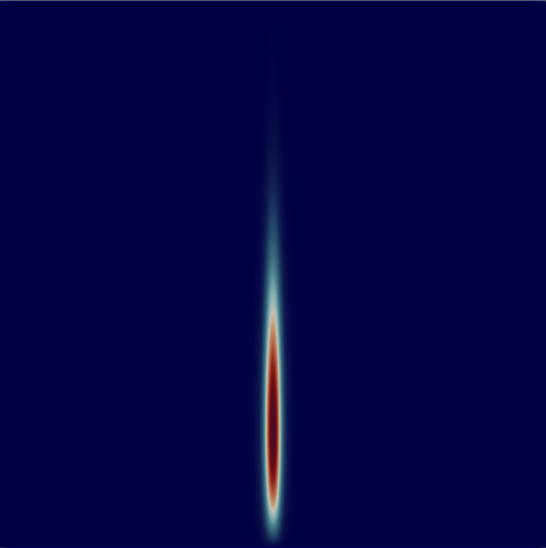}
\includegraphics[width=0.20\textwidth]{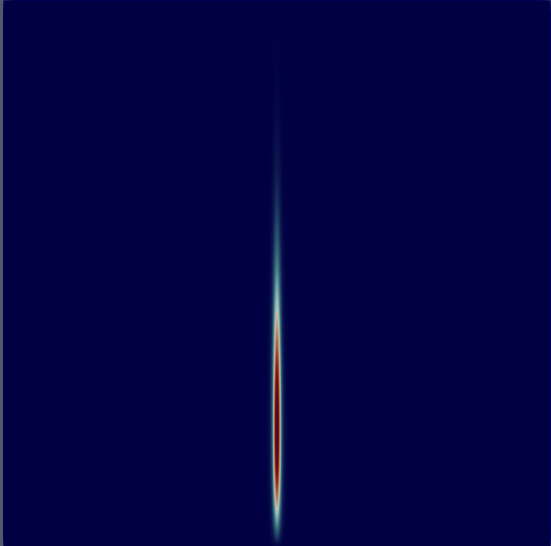}
\caption{Plot of the magnitude of the $L^2$ normalized first eigenfunction of the operator $\mathscr{L}_h$ with $V = x^2$ and $\alpha = 1 - 0.1x^2$, truncated to a square domain $[-R,R] \times [0,2R]$ ($R = 4$) with Dirichlet boundary conditions on the boundary, and for $h = 2^{-n}$ with $n = 4,6,8,10$ (from top-left to bottom-right). The red intensity is proportional to the magnitude of the plotted eigenfunction. Numerical computations were performed using the finite element method with a refined mesh near near the axis $\{x = 0\}$. The $L^2$ norm of the eigenfunction is small outside a box of size $O(h^{1/2})$ along the $x$-direction and $O(1)$ along the $y$-direction. The source code for these computations is available at \cite{code}.}
\label{fig:numericsLh}
\end{figure}

\begin{figure}[htbp]
\centering
\includegraphics[height=3cm]{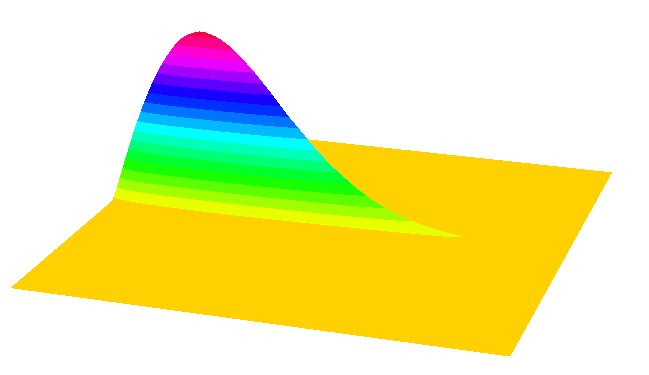} \includegraphics[height=3cm]{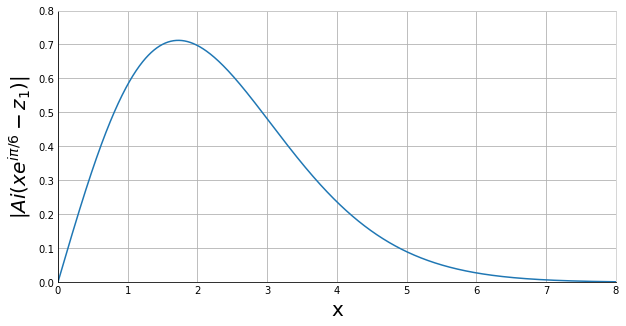}
\caption{Graph of the magnitude of the eigenfunction computed numerically as in Figure \ref{fig:numericsLh} for $n = 10$ (left) and  graph of the function $x \mapsto |\textup{Ai}(xe^{i \frac\pi6} + z_1)|$ (right).}
\label{fig:Airy}
\end{figure}

\section{About the main result and its proof}
Let us now describe our main result and the strategy of its proof.

\subsection{Statement of the main result}

Let $V, \alpha : \R \to \R$ be smooth, bounded functions. For $h > 0$, let 
$$\mathscr{L}_h : L^2(\R \times\R_+) \to L^2(\R \times \R_+)$$ 
be the (unbounded) linear operator defined by 
$$\mathscr{L}_h := h^{2/3} (D_y^2 + i\alpha(x) y) + (hD_x)^2 + iV(x)$$
on the domain
\begin{equation}
\label{eq:defH2}
\cHd := \Big\{\psi \in H^2(\R \times \R_+) \cap H^1_0(\R \times \R_+)\,\mid\, y \psi \in L^2(\R\times \R_+)\Big\}.
\end{equation}
We introduce the following additional assumptions on $V$ and $\alpha$.
\begin{assumption}[Assumptions on $V$]\label{ass:V}
The function $V$ satisfies the following properties 
\begin{itemize}
\item[(i)] $V$ is smooth and bounded as well as all of its derivatives. 
\item[(ii)] $V(x) \geq 0$ for all $x \in \R$ and $V(x)$ vanishes only for $x = 0$. This minimum is furthermore non-degenerate, i.e. $V''(0) > 0$.
\item[(iii)] $V$ is bounded below at infinity, i.e., $V_\infty := \displaystyle\liminf_{|x|\to \infty} V(x) > 0$. 
\end{itemize}
\end{assumption}
\begin{assumption}[Assumptions on $\alpha$]\label{ass:alpha}
The function $\alpha$ satisfies the following properties
\begin{itemize}
\item[(i)] $\alpha$ is smooth and bounded as well as all of its derivatives, and satisfies $\alpha \geq \alpha_0$ for some $\alpha_0 > 0$. 
\item[(ii)] $\alpha'(0) = 0$,
\item[(iii)] It holds that
\begin{equation}
\label{eq:cond_alpha}
(\inf\alpha)^{\frac23}|z_2|-(\sup\alpha)^{\frac23}|z_1|>0
\end{equation}
where $z_1 \approx -2.33811$ and $z_2 \approx -4.08795$ are the the rightmost real zeros of the Airy function.
\end{itemize}
\end{assumption}
For each $x \in \R$, we denote by $\lambda_{1,\alpha}(x)$ the ``first'' eigenvalue (i.e., smallest in magnitude) of the complex Airy operator $D_y^2 + i\alpha(x)u$ on the half line with Dirichlet conditions, which is given by 
\begin{equation}
\label{eq:def_lambda0}
\lambda_{1,\alpha}(x) := \alpha(x)^{2/3} |z_1| e^{i \frac{\pi}{3}}.
\end{equation}
Our main result below is an Agmon estimate (after S. Agmon, see \cite{agmon1982lectures}) describing the localization in the $x$-variable of the eigenfunctions associated to the eigenvalues in $D(\lambda_{1,\alpha}(0)h^{\frac23},Rh)$ in terms of the following Agmon distance
\begin{equation}
\label{eq:def_phieta}
\phi_{\mu}(x) := \frac{1-\mu}{\sqrt{2}}\left|\int_0^x\sqrt{V(s)}\,ds\right|
\end{equation}
where $\mu \in (0,1)$. 
\begin{theorem}
\label{thm:main}
	Let $V$ and $\alpha$ satisfy Assumptions \ref{ass:V} and \ref{ass:alpha}, let $R > 0$, $\mu\in(0,1)$ and let $\phi_\mu$ be defined by \eqref{eq:def_phieta}. There exist $C, h_0>0$ such that the estimate 
	\begin{equation}
	\label{eq:AgmonMain}
	\|e^{\phi_\mu/h}\psi\|_{L^2(\R \times \R_+)}\leq C\|\psi\|_{L^2(\R\times \R_+)}
	\end{equation}
	holds for all $h \in (0,h_0)$, $\lambda\in D\left(\lambda_{1,\alpha}(0)h^{\frac23},Rh\right)$ and all $\psi \in \mathcal{H}^2$ satisfying ${(\mathscr{L}_h - \lambda) \psi = 0}$.
	\end{theorem}

\begin{remark}
\label{rem:comments_main}
~
		\begin{enumerate}[\hspace{-18pt}\rm (i)]
	\item  Under Assumption \ref{ass:V}, $\phi_\mu(x) \sim \frac{1 - \mu}{2}\sqrt{V''(0)}x^2$ as $x \to 0$. Thus, roughly speaking, Theorem~\ref{thm:main} states that the eigenfunctions are exponentially localized in a $O(h^{1/2})$-neighborhood of the (half) $y$-axis.
	
	\item Theorem \ref{thm:main} is sharp, in the sense that the result does not hold if we allow $\mu = 0$ (in particular, under the same assumptions, one cannot prove a localization at a scale $O(h^{1/2-\varepsilon})$ for any $\varepsilon > 0$). This is shown in Section \ref{sec:sharp}.
	
	\item The condition (iii) in Assumption \ref{ass:alpha} is merely included to make the proof simpler, but Theorem \ref{thm:main} holds without it. The reason is explained in the sketch of the proof in Section \ref{sec:sketch}, see Remark \ref{rem:get_rid_condition}.

	\item \label{remitem:pollution} It is possible to obtain a  localization at scale $O(h^{1/3})$, by using simple variational arguments. The reason why this approach does not obtain the optimal $O(h^{1/2})$ scale is related to a ``pollution effect'' caused by the numerical range of the complex Airy operator. We explain this in more details in Appendix \ref{app:pollution}.

	\item The results of this paper can be used to show the existence of an eigenvalue of $\mathscr{L}_h$ in the disk $D(\lambda_{1,\alpha}(0)h^{\frac23},Rh)$ for $R > 0$ large enough --  in particular, Theorem \ref{thm:main} has non-empty assumptions.  This is shown in Appendix \ref{app:sp_non_vide}. 
	The analog of this result is known for Bloch-Torrey operators on bounded domains, and we thus recover it in our setting by a new method. It can even be shown using analytic dilation arguments that the disk $D(\lambda_{1,\alpha}(0)h^{\frac23},Rh)$ in fact contains the eigenvalue of $\mathscr{L}_h$ that is smallest in magnitude, but we do not prove this here for the sake of conciseness.

	\item Although Theorem \ref{thm:main} is stated for Dirichlet boundary conditions, it is not difficult to check that the proof in this paper extends to Neumann or Robin condition: the main difference is that the zeros of the Airy function must be replaced by the zeros of its derivative (for Neumann conditions) or of a generic linear combination of the Airy function and its derivative (for Robin boundary conditions).
	\item Our proof relies on a one-dimensional Agmon estimate in the range of a projection operator related to the first Airy eigenfunction, and an operator-valued semiclassical pseudodifferential argument in the complement of this space, which seems to be new in this context. We explain this in the next section.
	\end{enumerate}
	\end{remark}

\subsection{Sketch of the proof}

We start by reducing the proof to an elliptic estimate for an exponential conjugation of the operator, namely,
\begin{equation}
\label{eq:ellip_intro}
\|u\| \lesssim h^{-1}\|(\mathscr{L}_h^\Phi - \lambda) u\| + \|\mathbf{1}_{\{|x|\lesssim h^{1/2}\}} u\|
\end{equation}
for all $u \in \mathcal{H}^2$ and $\lambda \in D(\lambda_{1,\alpha}(0)h^{2/3},Rh)$, where $\Phi(x)$ is a suitable weight function and $\mathscr{L}_h^\Phi := e^{\Phi/h} \mathscr{L}_h e^{-\Phi/h}$ (see Proposition \ref{prop:main_ellip} below). Indeed, choosing $u = e^{\Phi/h} \psi$, the first term in the right-hand side vanishes, and the Agmon estimate follows by using that $e^{\Phi/h}$ is $h$-independently bounded on the set $\{|x| \lesssim h^{1/2}\}$. To simplify the explanation, we now sketch our method to obtain \eqref{eq:ellip_intro} in the particular case where $\Phi = 0$; the general case only involves minor modifications, mainly due to the fact that $\mathscr{L}_h^\Phi$ and $\mathscr{L}_h$ differ only by an operator of order one. 

The key idea is to regard the complex Schrödinger part $h^2 D_x^2 + iV(x)$ as a perturbation of the eigenvalue $\lambda$, thus viewing $\mathscr{L}_h-\lambda$ as
$$\mathscr{L}_h -\lambda  = h^{2/3} \mathscr{A}_\alpha -  \widetilde{\lambda}(x,hD_x)\,, $$
where 
$\mathscr{A}_\alpha = D_y^2 + i \alpha(x) y$ and $\widetilde{\lambda}(x,\xi) := \lambda - (\xi^2 + iV(x))$. This suggests that, roughly speaking, a bound on $(\mathscr{L}_h - \lambda)^{-1}$ could be obtained from an estimate of the ``operator-valued resolvent'' $(h^{2/3}\mathscr{A}_\alpha - \widetilde{\lambda})^{-1}$. One way to formalize this idea is to use the concept of pseudodifferential operators with {\em operator-valued symbols} (the main results needed here are recapped in \S \ref{sec:recap_pseudo}). Namely, we rewrite 
$$\mathscr{L}_h - \lambda = \textup{Op}_h (h^{2/3}\mathscr{A}_\alpha(x) - \tilde{\lambda}(x,\xi))$$ 
where for each $x \in \R$, $\mathscr{A}_\alpha(x) := D_y^2 + i\alpha(x) y$ is a one-dimensional differential operator acting on a subspace of $L^2(\R_+)$ with Dirichlet boundary conditions at $y = 0$. The {\em quantization} $\textup{Op}_h$ works analogously to a semiclassical quantization on scalar symbols. The operator $\mathscr{A}_\alpha(x)$ is well-understood: its spectrum consists of the simple eigenvalues given by 
$$\lambda_{n,\alpha}(x) = \alpha(x)^{2/3} |z_n| e^{i \frac{\pi}{3}}\,,$$
where $z_n$ is the $n$-th real zero of the Airy function. This leads to the following natural question: does the region
$$U := \left\{h^{-2/3} \widetilde{\lambda}(x,\xi) \,\,\big|\,\, h > 0\,, \,\, \lambda \in D(\lambda_{1,\alpha}(0)h^{\frac23},Rh)\,,\,\, (x,\xi) \in \R^2\right\}\subset \mathbb{C}$$
contain any poles of $(\mathscr{A}_\alpha(x) - z)^{-1}$? The answer is that, since $\xi^2 \geq 0$ and $V(x) \geq 0$, the perturbation $-(\xi^2 + iV)$ ``pushes'' the eigenvalue $\lambda$ towards the ``south-west'' of the complex plane. In particular, $U$ is contained in a diagonal half-plane $\mathcal{P}$ as represented in Figure~\ref{fig:half_plane_intro}. Assuming that for all $x \in \R^d$, $\lambda_{2,\alpha}(x)$ remains sufficiently larger than $\lambda_{1,\alpha}(0)$ for all $x$, the only possible pole of $(h^{2/3} \mathscr{A}_\alpha(x) - z)^{-1}$ in $U$ is $\lambda_{1,\alpha}(x)h^{2/3}$. The previous condition is guaranteed for $h$ small enough provided that
$$(\inf \alpha)^{2/3} |z_2| > (\sup \alpha)^{2/3} |z_1|,$$
i.e., if $\alpha$ satisfies the condition (iii) in Assumption \ref{ass:alpha}.

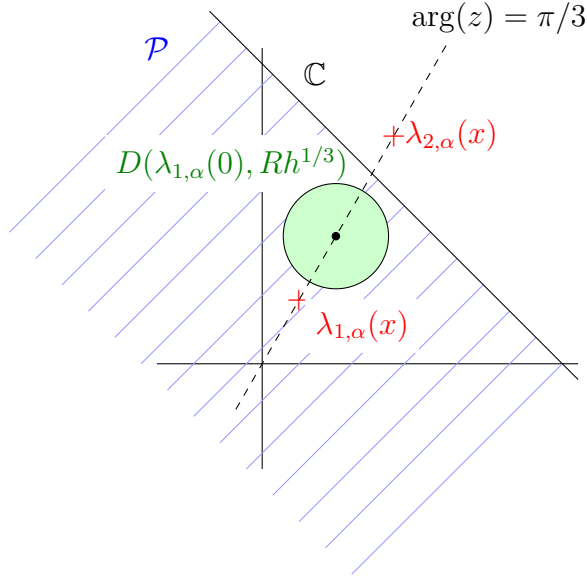
\begin{figure}[htbp]
\centering
\begin{tikzpicture}[scale=0.7]
\draw (-2,0)--(6,0);
\draw (0,-2)--(0,6);
\node at (1,5.5) {$\mathbb{C}$};
\coordinate (z1) at ({1.4*cos(60)},{1.4*sin(60)});
\coordinate (z2) at ({5*cos(60)},{5*sin(60)});
\coordinate (z10) at ({2.8*cos(60)},{2.8*sin(60)});

\node[above] at ({7*cos(60)+1},{7*sin(60)}) {$\arg(z) = \pi/3$};
\node[color=blue] at (-2,6) {$\mathcal{P}$};
\foreach \i in {0,...,13}
{\coordinate(a) at ($(5.7,0)+(-0.5*\i,0.5*\i)$);
\coordinate(b) at ($(a) + (-4,-4)$);
\draw[color=blue!40!white] (a)--(b);}
\node[green!50!black,fill=white] at (-0.58,3.8) {$D(\lambda_{1,\alpha}(0),Rh^{1/3})$};
\draw[fill=green!20!white] (z10) circle(1cm);
\draw[fill=black] (z10) circle(2pt);
\node[red] at (z2) {$+$};
\node[red,right] at (z2) {$\lambda_{2,\alpha}(x)$};
\node[red] at (z1) {$+$};
\node[red,fill=white] at ($(z1)+(1.2,-0.5)$) {$\lambda_{1,\alpha}(x)$};
\draw[dashed] ({-1*cos(60)},{-1*sin(60)})--({7*cos(60)},{7*sin(60)});
\draw[domain=-1:6] plot ({\x},{5.7 - \x});
\end{tikzpicture}
\caption{Poles of the resolvent $(\mathscr{A}_\alpha(x) - z)^{-1}$ (red), a region $\mathcal{P} \subset \mathbb{C}$ (blue hatched half plane) enclosing the set $U$ of the values taken by $h^{-2/3} \widetilde{\lambda}(x,\xi)$ for $(x,\xi) \in \R^2$ and $\lambda \in D(\lambda_{1,\alpha}(0)h^{2/3},Rh)$}
\label{fig:half_plane_intro}
\end{figure}

To capture the influence of the pole $\lambda_{1,\alpha}(x)$, it is useful to introduce the spectral projection $\pi_{1,\alpha}(x)$ onto to the one-dimensional eigenspace associated to $\lambda_{1,\alpha}(x)$. The point is that $h^{2/3}\mathscr{A}_\alpha(x) - \widetilde{\lambda}(x,\xi)$ can be boundedly inverted on the range of $\Id - \pi_{1,\alpha}(x)$, and the inverse $r(x,\xi)$ is itself a sufficiently well-behaved symbol (see Proposition \ref{prop:rlambda_class}). The quantization of this inverse will give a first-order parametrix for $\mathscr{L}_h^\Phi - \lambda$. Indeed, using a composition theorem for pseudodifferential operators with operator-valued symbols (Theorem \ref{thm:compo} below), we obtain
$$\textup{Op}_h(r(x,\xi)) \textup{Op}_h(h^{2/3} \mathscr{A}_\alpha(x) - \widetilde{\lambda}(x,\xi)) = \Id - \textup{Op}_h(\pi_{1,\alpha}) + \textup{remainders},$$
In other words, writing $R := \textup{Op}_h(r(x,\xi))$ and $\Pi_{1,\alpha} := \textup{Op}_h(\pi_{1,\alpha})$, we have
$$R(\mathscr{L}_h - \lambda) = \Id - \Pi_{1,\alpha} + \textup{remainders}.$$
(see \eqref{eq:composition} below). We then use the Calder\'{o}n-Vaillancourt theorem (Theorem \ref{thm:CadlderonVaillancourt} below) to show that $\|R\| \lesssim h^{-2/3}$ as well as to control the remainder terms. This leads to the estimate
\begin{equation}
\label{eq:ellip2_intro}
\|(\Id - \Pi_{1,\alpha})u\| \lesssim h^{-2/3} \|(\mathscr{L}_h - \lambda)u\| + h^{1/3}\|u\|,
\end{equation}
(see Proposition \ref{prop:ellip2}). 

\begin{remark}[Removing the assumption on $\alpha$]
\label{rem:get_rid_condition}
To avoid using the condition (iii) in Assumption \ref{ass:alpha}, one should instead invert the symbol of $\mathscr{L}_h - \lambda$ in the range of $\Id - \pi_{1,\alpha}(x) - \ldots - \pi_{n,\alpha}(x)$, where $\pi_{i,\alpha}(x)$ is the spectral projection on to the eigenspace associated to $\lambda_{i,\alpha}(x)$, and where $n$ is large enough so that $(\inf \alpha)^{2/3} |z_{n+1}| > (\sup \alpha)^{2/3} |z_1|$. 
\end{remark}

To complete the proof of the estimate \eqref{eq:ellip_intro}, it remains to estimate the contribution of $\Pi_{1,\alpha} u$. For this, we notice that 
$$(\mathscr{L}_h- \lambda)\Pi_{1,\alpha} = \left(\lambda_{1,\alpha}(x)h^{2/3} + (hD_x)^2 + iV - \lambda\right) \Pi_{1,\alpha};$$
that is, on the range of $\Pi_{1,\alpha}$, the operator $\mathscr{A}_\alpha$ is just the pointwise multiplication by $\lambda_{1,\alpha}(x)$. This effectively removes the pollution problem alluded to in Remark \ref{rem:comments_main} (\ref{remitem:pollution}), allowing to show the estimate
\begin{equation}
\label{eq:ellip1_intro}
\|\Pi_{1,\alpha}u\| \lesssim h^{-1}\|(\mathscr{L}_h - \lambda)\Pi_{1,\alpha}u\| + \|\mathbf{1}_{\{|x|\lesssim h^{1/2}\}}\Pi_{1,\alpha}u\|
\end{equation}
via simple variational arguments for a one-dimensional complex Schrödinger operator (see Proposition \ref{prop:ellip1}). The estimate \eqref{eq:ellip_intro} is finally obtained by summing \eqref{eq:ellip1_intro}, \eqref{eq:ellip2_intro}, combined with an estimate of the commutator $[\mathscr{L}_h,\Pi_{1,\alpha}]$ (Corollary \ref{cor:commut}).

\subsection{Organization of the paper}
\label{sec:sketch}

The remainder of this article is organized as follows. Section \ref{sec:reduction} reduces the proof of Theorem \ref{thm:main} to an elliptic estimate (Proposition \ref{prop:main_ellip}), which, after introducing a projection operator $\Pi_{1,\alpha}$, is further split into two key estimates, respectively on the range of $\Pi_{1,\alpha}$ and $(\Id - \Pi_{1,\alpha})$. The first one follows immediately from the properties of $\Pi_{1,\alpha}$ and is presented in \S \ref{sec:two_ellip}. Moreover, the (short) proof that 
Theorem \ref{thm:main} follows from the main elliptic estimate
can be found in \S\ref{sec:main_ellip}. Section \ref{sec:2D} -- the bulk of the paper -- is devoted to the second and more subtle key estimate on the range of $\Id - \Pi_{1,\alpha}$. Finally, in Section \ref{sec:sharp}, we establish the sharpness of Theorem \ref{thm:main}.

We also include several appendices. Appendix \ref{app:density} proves a density result used in the proofs of Section \ref{sec:reduction}. Appendix \ref{app:er} gathers elliptic regularity estimates for the operators involved in the paper. Appendix \ref{app:sp_non_vide} shows that our method gives the existence of an eigenvalue of $\mathscr{L}_h$ in the disk appearing in the assumption of Theorem \ref{thm:main}. Finally, Appendix \ref{app:pollution} discusses the ``pollution'' by the Airy operator mentioned in Remark \ref{rem:comments_main} (\ref{remitem:pollution}).

\subsection{Notation} Here we gather some of the notation used throughout this article. Let $C^\infty_c(\R^d)$ denote the space of infinitely differentiable functions in $\R^d$ with compact support, and $\mathscr{S}(\R^d)$ the Schwartz space. For a closed subset $F \subset \R^d$, $\mathscr{S}(F)$ denotes the set of restrictions to $F$ of elements of $\mathscr{S}(\R^d)$. For $n \in \mathbb{N}$ and $X \in \R^n$, let
$$\langle X \rangle := (1 + |X|^2)^{1/2}.$$
Let $D = \frac{1}{i}\partial$ (i.e., $D_{x_i} = \frac{1}{i} \frac{\partial}{\partial_{x_i}}$ and similarly for $D_{\xi_i}$ and so on). We denote by $\|\cdot\|_\infty$ the supremum norm on $\R^d$. 

We write $\mathscr{L}(E,F)$ for the space of bounded linear maps between the normed spaces $E$ and $F$ with the norm $$\|A\|_{\mathscr{L}(E,F)} := \sup_{x \in E\setminus \{0\}} \frac{\|Ax\|_{F}}{\|x\|_{E}}.$$ 
When $E = F$ (with equal norms), we put $\mathscr{L}(E):=\mathscr{L}(E,E)$ and let $\|\cdot\|_{\mathscr{L}(E)}$ denote the corresponding norm. As usual, the adjoint of a (possibly unbounded) operator $A$ is denoted by $A^*$, and the {\em transpose} $A^T$ is the operator defined on $\textup{dom}(A^*)$ by 
$$A^Tu:= \overline{A^* \overline{u}}.$$ 
Given two operators $A,B$, we denote by $[A,B] = AB - BA$ their commutator, whenever it makes sense. Given a Hilbert space $H$ and an (unbounded) linear map $A : H \to H$, we denote by $\sigma(A)\subset \mathbb{C}$ the spectrum of $A$ and by $\rho(A) \subset \mathbb{C}$ its resolvent set. We denote by $D(z_0,r)$ the open disk centered at $z_0 \in \mathbb{C}$ and with radius $r > 0$, and by $\mathscr{C}(z_0,r)$ its boundary, with the counter-clockwise orientation when used in contour integrals. 

Let $\cH := L^2(\R \times \R_+)$. We denote by $\|\cdot\|_{\R \times \R_+}$ and $\langle \cdot,\cdot \rangle_{\R \times \R_+}$ its usual norm and inner product, and we will often drop the $\R \times \R_+$ subscript when it will not lead to confusion. Let $\cH^2$ be as in \eqref{eq:defH2} with the norm
$$\|\psi\|_{\mathcal{H}^2}^2 := \|\psi\|^2 + \|\Delta \psi\|^2 + \|y \psi\|^2.$$
We also consider the space 
$$\mathrm{D} = \big\{f \in H^2(\R_+) \cap H^1_0(\R_+) \,\mid\, y f\in L^2(\R)\big\},$$
endowed with the norm 
$$\|f\|_{\mathrm{D}}^2 := \|f\|_{\R_+}^2 + \|\mathscr{A}f\|_{\R_+}^2\,,$$ 
where $\mathscr{A} := D_y^2 + iy$ and $\|\cdot\|_{\R_+}$ denotes the $L^2$ norm on $\R_+$. We denote by $\mathscr{A}_\alpha$ the operator defined on $\cHd$ by 
\begin{equation}
\mathscr{A}_\alpha  := D_y^2  + i\alpha(x)y
\end{equation}
(i.e., $\mathscr{A}_\alpha$ acts on functions of the variables $x$ and $y$) and for each $x \in \R$, we write $\mathscr{A}_\alpha(x)$ for the operator defined by the same formula but acting on $\mathrm{D}$
(i.e., for each $x$, $\mathscr{A}_\alpha(x)$ acts on functions of the variable $y$). 

Let $\mathrm{Ai}$ be the standard Airy function which can be defined for $x \in \R$ by the semi-convergent integral
$$\mathrm{Ai}(x) = \frac{1}{\pi} \int_{0}^\infty \cos\left(\frac{t^3}3 + xt\right)dt\,,$$
and let $0 > z_1 > z_2 > \ldots$ be the sequence of its zeros ($z_1 \approx -2.33811$, $z_2 \approx -4.08795$, $z_3 \approx -5.52056$, etc.). For every $x\in \R$ and $n \in \mathbb{N}$, we write
\begin{equation}
\label{eq:def_lambdah(x)}
\lambda_{n,\alpha}(x) := \alpha(x)^{2/3} |z_n| e^{i \frac{\pi}{3}}.
\end{equation}

To lighten the proofs, we will often denote by $C$ any generic positive constant whose value can be bounded independently of the universally quantified variables in the statement.

\section{Reduction to elliptic estimates}

\label{sec:reduction}
\subsection{Elliptic estimate for a conjugate operator}

\label{sec:main_ellip}
In this section, we reduce the proof of Theorem \ref{thm:main} to a global elliptic estimate (Proposition \ref{prop:main_ellip} below) for a conjugate operator of the form
\begin{equation}
\label{eq:def_Lphih}
\mathscr{L}^\Phi_h := e^{\Phi/h} \mathscr{L}_h e^{-\Phi/h} = h^{2/3} \mathscr{A}_\alpha + (hD_x + i\Phi')^2 + iV.
\end{equation}
Considering the leading order in $h$, one can guess that the eikonal equation $\phi(x)'^2 = iV(x)$ should play a role. Namely, here we consider weights $\Phi$ in \eqref{eq:def_Lphih} that are controlled by the real part of the eikonal solution $\phi$ in the following sense.
\begin{definition}[$\mu$-subsolution]
\label{def:etaSubsolution}
Let $\mu \in (0,1)$ and let $\Phi : \R \to \R_+$ be an infinitely differentiable function which is bounded as well as all of its derivatives.

We say that $\Phi$ is a {\em $\mu$-subsolution} 
if it satisfies
$$V - 2\Phi'^2 \geq \mu V.$$
\end{definition}

\begin{proposition}[Main elliptic estimate]
\label{prop:main_ellip}
For any $R > 0$ and $\mu \in (0,1)$, there exists $C(R,\mu) > 0$, $h_0 > 0$, $L > 0$ and $N > 0$ such that the inequality
$$\|u\| \leq C(R,\mu) \max_{n \leq N} \|\partial_x^n\Phi'\|_{\infty} \left(h^{-1}\|(\mathscr{L}_h^\Phi - \lambda)u\| + \|\mathbf{1}_{\{|x| \leq Lh^{1/2}\}} u\|\right)$$
holds for any $h \in (0,h_0)$, $\lambda \in D(\lambda_{1,\alpha}(0)h^{\frac23},Rh)$, $u \in \cHd$ and for any $\mu$-subsolution $\Phi$. 

\end{proposition}  
This result is obtained as a direct consequence of Propositions \ref{prop:ellip1} and \ref{prop:ellip2} below. We now show that Theorem \ref{thm:main} follows from Proposition \ref{prop:main_ellip}.
\begin{proof}[Proof of Theorem \ref{thm:main} using Proposition \ref{prop:main_ellip}]
For $A \geq 1$, consider $\Phi_{A}(x) := \frac{1 - \mu}{\sqrt{2}} \left|\int_{0}^x \sqrt{V_A(s)}\,ds\right|$ with $V_A = V \chi_A$, where $\chi_A = \chi(x/A)$ where $\chi \geq 0$, $\supp \chi \subset [-1,1]$ and $\chi \equiv 1$ near $0$. Then $\Phi_A$ is a $\mu$-subsolution. Moreover, $\Phi_A \leq \phi_\mu$ and the derivatives of $\Phi'_A$ are uniformly bounded (in terms of $V$, $\chi$, and $A$); hence by Proposition \ref{prop:main_ellip} applied to $u := e^{\Phi_A/h} \psi$,
$$\|e^{\Phi_A/h} \psi\| \leq C(R,h_0) \|\mathbf{1}_{\{|x|\leq Lh^{1/2}\}} e^{\Phi_A/h} \psi\| \leq C(R,h_0) \|\mathbf{1}_{\{|x|\leq Lh^{1/2}\}}e^{\phi_\mu/h} \psi\|$$
where, importantly, $C(R,h_0)$ is independent of $A$. Using the Fatou lemma, we conclude that
$$\|e^{\phi_\mu/h} \psi\| = \liminf_{A \to \infty}\|e^{\Phi_A/h} \psi\| \leq C(R,h_0) \|\mathbf{1}_{\{|x|\leq Lh^{1/2}\}} e^{\phi_\mu/h} \psi\|.$$
The result follows using that $\mathbf{1}_{\{|x|\leq Lh^{1/2}\}} \phi_\mu$ can be bounded independently of $h$, since $\phi_\mu$ behaves quadratically near $x = 0$.
\end{proof}

\subsection{The projection \texorpdfstring{$\Pi_{1,\alpha}$}{Pi\_1,alpha} and its properties}
\label{sec:Pialpha}
Proposition \ref{prop:main_ellip} will be proved separately on the range and the kernel of projection operator $\Pi_{1,\alpha}$. We now define this operator and establish some of its properties. Let $\mathscr{A}$ be the {\em complex Airy operator on the half-line}, i.e., the unbounded linear operator defined by 
$$\mathscr{A} := D_y^2 + iy$$
densely defined on the domain
\begin{equation}
\label{eq:def_D}
\mathrm{D}:=\{f\in H^2(\R_+)\cap H^1_0(\R_+)\ : yf\in L^2(\R_+)\}.
\end{equation}
This operator has been studied in \cite{grebenkov2017complex,savchuk2017spectral}. Its adjoint $\mathscr{A}^*$ is the unbounded linear operator defined by 
$\mathscr{A}^* = D_y^2 - iy$
with domain $\mathrm{D}$; thus, 
\begin{equation}
\label{eq:adjointA}
\mathscr{A}^T = \mathscr{A}.
\end{equation}
For any $x \in \R$, let $\mathscr{A}_\alpha(x) : \mathrm{D} \to L^2(\R_+)$ be defined by 
\begin{equation}
\label{eq:def_Ax}
\mathscr{A}_\alpha(x) := D_y^2 + i \alpha(x)y.
\end{equation}
Observe that
\begin{equation}
\label{eq:Ax_is_UxAUx*}
\forall x \in \R\,, \quad \mathscr{A}_\alpha(x)= \alpha(x)^{2/3}\mathcal{U}_\alpha(x) \mathscr{A} \mathcal{U}_\alpha(x)^*
\end{equation}
where $\mathcal{U}_\alpha(x) : L^2(\R_+) \to L^2(\R_+)$ is the unitary operator
\begin{equation}
\label{eq:defUx}
(\mathcal{U}_\alpha(x) f)(y) = \alpha(x)^{1/6} f(\alpha(x)^{1/3}y).
\end{equation} 
It is well-known that $\mathscr{A}$ is closed (see, e.g., \cite[Proposition 3]{savchuk2017spectral}), and 
thus, so is $\mathscr{A}_\alpha(x)$ for each $x$. In what follows, we equip $\mathrm{D}$ with the graph norm
\begin{equation}
\label{eq:def_normeD}
\|f\|^2_{\mathrm{D}} := \|f\|^2_{\R_+} + \|\mathscr{A}f\|^2_{\R_+}.
\end{equation}

The operator $\mathscr{A}$ has compact resolvent and its spectrum is given by
\[\sigma(\mathscr{A})=\left\{|z_n|e^{i\frac{\pi}{3}}\,\mid\, n= 1,2,\ldots\right\}\,\]
where $0 > z_1 > z_2 > \ldots$ is the sequence of the real zeros of the Airy function. Let $\pi_1 : L^2(\R_+) \to L^2(\R_+)$ be the Riesz projector 
\begin{equation}
\label{eq:def_Pi1}
\pi_1 := \frac{1}{2\pi i} \int_{\mathscr{C}(|z_1|e^{i\frac{\pi}{3}},r)} (z - \mathscr{A})^{-1}\,dz
\end{equation}
where $r < |z_2| - |z_1|$.
\begin{proposition}
\label{prop:u1}
There exists $u_1 \in \mathrm{D}$ satisfying $\mathscr{A} u_1 = |z_1| e^{i \frac{\pi}{3}} u_1$ and $\langle u_1,\overline{u}_1\rangle_{L^2(\R_+)} = 1$ (where $\langle \cdot ,\cdot\rangle_{\R_+}$ denotes the standard inner product of $L^2(\R_+)$). Moreover,
$$\pi_1 f = \langle f,\overline{u_1}\rangle_{\R_+} u_1 \qquad \textup{for all } f \in L^2(\R_+).$$
\end{proposition}
\begin{proof}
~
\begin{enumerate}[1.]
\item We first observe that for $z \in \gamma$, since $\mathscr{A}$ is closed and densely defined, $((z - \mathscr{A})^{-1})^* = (z- \mathscr{A}^*)^{-1}$, thus
$$\pi_1^* = \frac{1}{2\pi i} \int_{\gamma'} (z - \mathscr{A}^*)^{-1}\,dz$$
where $\gamma'$ is a circle centered at $|z_1| e^{-i\pi/3}$ and with radius $r < |z_2| - |z_1|$.

\item \label{step2_propu1}Let $u_1$ be an eigenvector of $\mathscr{A}$ associated to $|z_1| e^{i\frac{\pi}{3}}$. Observe that $\mathscr{A}^* \overline{u_1} = \overline{\mathscr{A}u_1} = |z_1|e^{-i\pi/3}\overline{u_1}$. Since the eigenvalues of $\mathscr{A}$ and $\mathscr{A}^*$ are algebraically simple, it follows (see, e.g., \cite[Proposition 3.35]{cheverry2021guide}) that $\textup{Ran}(\pi_1) = \textup{Span}(\{u_1\})$ and $\textup{Ran}(\pi_1^*) = \textup{Span}(\{\overline{u_1}\})$. Thus there exist $v,w \in L^2(\R_+)$ such that for all $f \in L^2(\R_+)$, 
$$\pi_1 f = \langle f,v\rangle_{\R_+} u_1\,, \quad \pi_1^* f = \langle f,w\rangle_{\R_+} \overline{u_1}.$$
In particular, $\textup{Ker}(\pi_1) = \{v\}^\perp$ and $\textup{Ran}(\pi_1^*) = \textup{Span}(\{\overline{u}_1\})$. Since $\textup{Ker}(\pi_1) = \textup{Ran}(\pi_1^*)^\perp$, it follows that $v = c\overline{u}_1$ for some $c \in \mathbb{C}$. 
\item Since $\pi_1 u_1 = u_1$, we deduce from step \ref{step2_propu1} that 
$$\langle u_1,c\overline{u}_1\rangle_{\R_+} = 1;$$
in particular, $c \neq 0$ and $\langle u_1,\overline{u}_1\rangle_{\R_+} = \frac{1}{\overline{c}}$. Choosing $c'$ such that $c'^2 = \overline{c}$, it is easy to check that the function $\tilde{u}_1 := c' u_1$ satisfies the requirements. \qedhere
\end{enumerate}
\end{proof}

In what follows, we fix $u_1$ as in Proposition \ref{prop:u1} and for all $x \in \R$, we denote
\begin{equation}
\label{eq:def_ux}
u_{\alpha,x}(y) := (\mathcal{U}_\alpha(x) u_1)(y).
\end{equation}
For all $x \in \R$, 
\begin{equation}
\label{eq:uxux*=1}
\langle u_{\alpha,x},\overline{u_{\alpha,x}}\rangle_{\R_+} = \langle \mathcal{U}_\alpha(x)^* \mathcal{U}_\alpha(x) u_1,\overline{u_1}\rangle_{\R_+} = 1;
\end{equation} 
Since $u_1 \in \mathscr{S}(\R_+)$, one can check that the map $\R \ni x \mapsto u_{\alpha,x} \in L^2(\R_+)$ is infinitely differentiable and satisfies
\begin{equation}
\label{eq:Dxux}
\forall k \in \mathbb{N}\,, \,\, \exists C_k > 0\,: \,\, \forall x \in \R\,, \quad \|D_x^k u_{\alpha,x}\|_{\R_+} \leq C_k.
\end{equation}
Let $j_\alpha : L^2(\R) \to \cH$ be defined for $f \in L^2(\R_+)$ by
$$(j_\alpha f)(x,y) := f(x)\cdot u_{\alpha,x}(y) \quad \textup{for all } (x,y) \in \R \times \R_+$$
and let $j_\alpha^T : \cH \to L^2(\R)$ be its transpose, i.e., 
\begin{equation}
\label{eq:def_jT}
\langle j_\alpha^T \psi,f\rangle_{\R_+} := \langle \psi,\overline{(j_\alpha\overline{f})}\rangle_{\R \times \R_+}\,, \quad (\psi,f) \in \cH \times L^2(\R_+).
\end{equation}
By the Fubini theorem, 
\begin{equation}
\label{eq:expression_jT}
(j_\alpha^T \psi)(x) = \langle \psi(x,\cdot),\overline{u}_{\alpha,x}\rangle_{\R_+}.
\end{equation}
\begin{definition}[The operator $\Pi_{1,\alpha}$]
We define $\Pi_{1,\alpha} : \cH \to \cH$ by 
$$\Pi_{1,\alpha} := j_\alpha \cdot j_\alpha^T.$$
\end{definition} 

To avoid worrying about the pointwise evaluation of almost-everywhere defined functions, it will be convenient to use the following density result. Its proof can be found in Appendix~\ref{app:density}.
\begin{lemma}
\label{lem:density}
The space 
$$X := C_{c}^\infty(\overline{\R\times \R_+}) \cap \cHd$$ 
is dense in $\cHd$. Here, $C_{c}^\infty(\overline{\R \times \R_+})$ denotes the set of restrictions to $\R \times \R_+$ of functions in $C_{c}^\infty(\R^2)$.  
\end{lemma}

\begin{proposition}[Elementary properties of $j_\alpha$ and $\Pi_{1,\alpha}$]
\label{prop:jpi_elem}
The maps $j_\alpha : L^2(\R) \to \cH$, $j_\alpha^T : \cH \to L^2(\R)$ and $\Pi_{1,\alpha} : \cH \to \cH$ are bounded and satisfy 
\begin{itemize}
\item[(i)] $j_\alpha^T  \cdot j_\alpha = \Id_{L^2(\R)}$ and $\Pi_{1,\alpha}$ is a projection (i.e., $\Pi_{1,\alpha}^2 = \Pi_{1,\alpha}$).
\item[(ii)] The embeddings $j_\alpha(H^2(\R)) \subset \cHd$, $j_\alpha^T(\cHd) \subset H^2(\R)$ and (thus) $\Pi_{1,\alpha}(\cHd) \subset \cHd$, hold and are continuous.
\item[(iii)] $\Pi_{1,\alpha} F = F \Pi_{1,\alpha}$ holds for any multiplication operator of the form $(Fu)(x,y) := F(x)u(x,y)$. 
\item[(iv)] $\Pi_{1,\alpha} \mathscr{A}_\alpha = \mathscr{A}_\alpha \Pi_{1,\alpha} = \lambda_{1,\alpha} \Pi_{1,\alpha}$ as operators from $\cHd$ to $L^2(\R_+ \times \R)$, where $\lambda_{1,\alpha}$ is the multiplication operator $(\lambda_{1,\alpha} u)(x,y) := \lambda_{1,\alpha}(x) u(x,y)$, with $\lambda_{1,\alpha}(x)$ defined by \eqref{eq:def_lambdah(x)}. 
\end{itemize}
\end{proposition}
\begin{proof}
~
\begin{itemize}
\item[(i)] From \eqref{eq:uxux*=1} and the Fubini theorem, we deduce that $\langle j_\alpha u,\overline{j_\alpha\overline{v}}\rangle_{\R \times \R_+} = \langle u,v\rangle_{\R}$, which implies that $j_\alpha^T \cdot j_\alpha = \Id_{L^2(\R)}$. Thus $\Pi_{1,\alpha}^2 = (j_\alpha\cdot j_\alpha^T)\cdot (j_\alpha \cdot j_\alpha^T)= j_\alpha\cdot (j_\alpha^T \cdot j_\alpha) \cdot j_\alpha^T = \Pi_{1,\alpha}$. 
\item[(ii)] The continuous embedding $j_\alpha(H^2(\R))$ follows from the definition of $j_\alpha$ and the property \eqref{eq:Dxux}. On the other hand, for $\psi \in X$, the combination of \eqref{eq:expression_jT} and \eqref{eq:Dxux} and differentiation under the integral sign implies that $j_\alpha^T \psi \in H^2(\R)$ with $\|j_\alpha^T \psi\|_{H^2(\R)} \leq \|\psi\|_{\cHd}$, and the continuous embedding $j_\alpha^T(\cHd) \subset H^2(\R)$ follows by density.  
\item[(iii)] By definition of $j_\alpha$, it is immediate that $j_\alpha F = Fj_\alpha$ (where $F$ stands both for the multiplication operator on $L^2(\R)$ and on $\cH$). It follows by taking the transpose that $j_\alpha^T F = F j_\alpha^T$, and thus $\Pi_{1,\alpha} F =( j_\alpha \cdot j_\alpha^T) F = F (j_\alpha \cdot j_\alpha^T) = F \Pi_{1,\alpha}$.
\item[(iv)] For any $f \in H^2(\R)$, 
$$\mathscr{A}_\alpha (j_\alpha f) = \lambda_{1,\alpha} (j_\alpha f)$$
by definition of $j_\alpha$ and using that for all $x \in \R$, $(D_y^2 + \alpha(x) iy )u_{\alpha,x} = \lambda_{1,\alpha}(x) u_{\alpha,x}$. By taking the transpose and using that $\mathscr{A}_\alpha^T = \mathscr{A}_\alpha$, it follows that $j^T \mathscr{A}_\alpha = j^T \lambda_{1,\alpha}$. Thus by (iii), 
$$\Pi_{1,\alpha} \mathscr{A}_\alpha f = j_\alpha \cdot j_\alpha^T \mathscr{A}_\alpha f = j_\alpha \cdot j_\alpha^T (\lambda_{1,\alpha} f) = \lambda_{1,\alpha} j_\alpha \cdot j_\alpha^T f = \lambda_{1,\alpha} \Pi_{1,\alpha} f.$$
This shows that $\Pi_{1,\alpha} \mathscr{A}_\alpha = \lambda_{1,\alpha} \Pi_{1,\alpha}$ and by taking the transpose and using again (iii), it follows that $\mathscr{A}_\alpha \Pi_{1,\alpha} = \lambda_{1,\alpha} \Pi_{1,\alpha}$. \qedhere
\end{itemize}
\end{proof}
\begin{proposition}[{Estimate of the commutators $[D_x^2,\Pi_{1,\alpha}]$}]
\label{prop:commutator}
There exists $C > 0$ such that for all $\psi \in \cHd$, 
$$\|[D_x^2,\Pi_{1,\alpha}]\psi\|_{\R \times \R_+} \leq C \|D_x\psi\|_{\R \times \R_+}.$$
\end{proposition}
\begin{proof}
For $\psi \in C_c^\infty(\overline{\R \times \R_+}) \cap \cHd$, one can check that
$$\Pi_{1,\alpha} \psi(x,y) = \langle \psi_x,\overline{u_{\alpha,x}}\rangle_{\R_+}u_{\alpha,x}.$$
In this case, the claimed inequality is then obtained by differentiating under the integral sign and using \eqref{eq:Dxux}.
The result follows for any $\psi \in \cHd$ using the continuity of $[D_x^2,\Pi_{1,\alpha}]$ from $\cHd$ to $\cH$ (by property (ii) of Proposition \ref{prop:jpi_elem}) and density (Lemma \ref{lem:density}).
\end{proof}

\begin{corollary}
\label{cor:commut}
For all $R > 0$, $\mu \in (0,1)$, and $h_0 > 0$, there exists $C > 0$ such that the estimate
$$\|[\mathscr{L}_h^\Phi,\Pi_{1,\alpha}]\psi\|_{\R \times \R_+} \leq Ch^{4/3} \|\psi\|_{\R \times \R_+} + Ch^{2/3}\|(\mathscr{L}^\Phi_h - \lambda) \psi\|_{\R \times \R_+}$$
holds for all $h \in (0,h_0)$, $\lambda \in D(\lambda_{1,\alpha}(0)h^{2/3},Rh)$, $\psi \in \cHd$ and $\Phi$ any $\mu$-subsolution in the sense of Definition \ref{def:etaSubsolution}. 
\end{corollary}
\begin{proof}
By Proposition \ref{prop:jpi_elem}, $[\mathscr{L}_h^\Phi,\Pi_{1,\alpha}] = h^2[D_x^2,\Pi_{1,\alpha}]$, and thus by Proposition \ref{prop:commutator}, 
$$\|[\mathscr{L}_h,\Pi_{1,\alpha}]\|_{\R \times \R_+} \leq Ch \|(hD_x)\psi\|_{\R \times \R_+}.$$
Using the basic elliptic estimate of Proposition \ref{prop:basic_ellip}, we deduce that for any $\varepsilon \in (0,1)$, 
$$\|[\mathscr{L}_h,\Pi_{1,\alpha}]\|_{\R \times \R_+} \leq C\varepsilon h\|(\mathscr{L}_h^\Phi- \lambda)\psi\|_{\R \times \R_+} + h\left(\varepsilon \lambda + \frac{1}{\varepsilon}\right)\|\psi\|_{\R \times \R_+}\,,$$
and the conclusion follows by taking $\varepsilon = h^{-1/3}$. 
\end{proof}

\subsection{The two key elliptic estimates}
\label{sec:two_ellip}

With the operator $\Pi_{1,\alpha}$ at hand, we can now state the two key elliptic estimates used to prove Proposition \ref{prop:main_ellip}. 

\begin{proposition}[Elliptic estimate on $\textup{Ran}(\Pi_{1,\alpha})$]
\label{prop:ellip1}
Let $R > 0$ and $\mu \in (0,1)$. Then, there exist $h_0 > 0$, $C(R,\mu,h_0) > 0$, and $L > 0$ such that the inequality
\begin{equation}
\|\Pi_{1,\alpha} \psi\|_{\R \times \R_+} \leq C(R,\mu,h_0) \left(h^{-1}\|(\mathscr{L}_h^\Phi - \lambda) \psi\|_{\R \times \R_+} + \|\mathbf{1}_{\{|x|\leq Lh^{1/2}\}} \psi\|_{\R \times \R_+} + h^{1/3}\|\psi\|_{\R \times \R_+} \right)
\end{equation}
holds for any $h \in (0,h_0)$, $\lambda \in D(\lambda_{1,\alpha}(0)h^{\frac23},Rh)$, $u \in \cHd$ and any $\mu$-subsolution $\Phi$. 
\end{proposition}
\begin{proof}
By Proposition \ref{prop:jpi_elem}, 
$$(\mathscr{L}_h^\Phi - \lambda) \Pi_{1,\alpha} = ([(hD_x + i\Phi')^2 + V_{h,\rm eff}]-z)\Pi_{1,\alpha}$$
on $\cH^2$, where $V_{h,{\rm eff}} = iV + h^{2/3} (\lambda_{1,\alpha}(x)- \lambda_{1,\alpha}(0))$ and $z = \lambda - \lambda_{1,\alpha}(0)h^{\frac23}$ satisfies $z \in D(0,Rh)$. 

Denoting $\mathscr{H}^\Phi_{\rm eff} := (hD_x + i\Phi')^2 + V_{h,{\rm eff}}$, we observe that for all $f \in H^2(\R)$, 
$$ \langle(\mathscr{H}_{\rm eff}^\Phi - z) f,f \rangle_\R = \|(hD_x)f\|^2_{\R} + 2i\Re \big(\big\langle (hD_x) f,\Phi' f\big\rangle_{\R}\big) + \langle (V_{h,{\rm eff}} - \Phi'^2-z)f,f \rangle_\R$$
where $\|\cdot\|_\R$ and $\langle \cdot,\cdot\rangle_\R$ stand for the norm and inner product of $L^2(\R)$, respectively. Multiplying by $e^{-i\frac\pi 4}$, taking the real part, using that 
$2\Re(\langle (hD_x)f ,\Phi' f\rangle_\R) \geq - \|(hD_x) f\|_{\R}^2 - \|\Phi' f\|_{\R}^2$ and $|z| \leq Rh$, we deduce that
\begin{align*}
&\textup{Re}\left[e^{-i\frac{\pi}{4}}\big\langle(\mathscr{H}_{\rm eff}^\Phi - z) f,f \big\rangle_\R \right] \\
&\qquad\geq \cos(\pi/4)\big\langle (V - 2 \Phi'^2 - Ch)f,f \big\rangle_{\R} + h^{2/3}\cos(\pi/12)|z_1| \big\langle \big(\alpha^{\frac23}- \alpha(0)^{\frac23}\big)f,f\big\rangle_\R\\
&\qquad \geq  c \big\langle \big[V + h^{\frac23}(\alpha^{\frac23}- \alpha(0)^{\frac23})-Ch\big]f,f \big\rangle_\R
\end{align*}
for some $c > 0$, since $\Phi$ is a $\mu$-subsolution. We then write
$$V(x) + h^{\frac23}(\alpha(x)^{\frac23}-\alpha(0)^{\frac23}) = \frac{V(x)}{2} +\left[ \frac{V(x)}{2} + h^{2/3}(\alpha(x)^{2/3}-\alpha(0)^{2/3})\right]\,,$$
and notice that, under Assumptions \ref{ass:V} and \ref{ass:alpha}, for $h$ small enough,  the term between square brackets in the right-hand side is non-negative for all $x \in \R$.
Therefore,
\begin{equation}
\label{eq:lower_bound_Heff}
\textup{Re}\left[e^{-i\frac\pi4}\big\langle(\mathscr{H}^\Phi_{\rm eff} - z)f ,f\big\rangle_\R\right]
\geq ch\|f\|_{\R}^2 - Ch\|\mathbf{1}_{\{V \leq Ch\}}f\|_{\R}^2.
\end{equation}

Given $y > 0$ and $\psi \in C^\infty_c(\overline{\R\times \R_+}) \cap \cHd$, we apply \eqref{eq:lower_bound_Heff} to $f = (\Pi_{1,\alpha} \psi)(\cdot,y) \in H^2(\R)$. Integrating in $y$ and using the Fubini theorem, we deduce that
$$\|\Pi_{1,\alpha} \psi\|_{\R \times \R_+}^2 \leq Ch^{-1}\textup{Re} \left[e^{-i\frac\pi4}\big \langle (\mathscr{L}_h - \lambda) \Pi_{1,\alpha} \psi ,\Pi_{1,\alpha} \psi\big\rangle_{\R \times \R_+}\right] + C \|\mathbf{1}_{\{V \leq Ch\}}\Pi_{1,\alpha} \psi\|_{\R \times \R_+}^2,$$
and the same follows for all $\psi \in \cH^2$ by density (Lemma \ref{lem:density}). Finally, using the Cauchy-Schwarz inequality, the fact that $\Pi_{1,\alpha}$ commutes with $\mathbf{1}_{\{V \leq Ch\}}$ (by Proposition \ref{prop:jpi_elem} (iii)), 
$$\|\Pi_{1,\alpha} \psi\|_{\R \times\R_+} \leq Ch^{-1} \|(\mathscr{L}_h^\Phi - \lambda) \Pi_{1,\alpha} \psi\|_{\R \times\R_+} + C \|\mathbf{1}_{\{V \leq Ch\}}\psi\|_{\R \times\R_+},$$
and the result follows by the commutator estimate of Corollary \ref{cor:commut} and the fact that $V(x) \sim \frac{V''(0)}{2}x^2$ for $x$ close to $0$. 
\end{proof}

\begin{proposition}[Elliptic estimate on $\textup{Ran}(\Id - \Pi_{1,\alpha})$]
\label{prop:ellip2}
Let $R > 0$ and $\mu > 0$. There exists $C(R,\mu) > 0$, $h_0 > 0$ and $N > 0$ such that the inequality
\begin{equation}
\|(\Id -\Pi_{1,\alpha}) \psi\|_{\R \times \R_+} \leq C(R,\mu)\max_{n \leq N} \|\partial_x^n \Phi'\|_\infty\big( h^{-\frac23}\|(\mathscr{L}_h^\Phi - \lambda) \psi\|_{\R \times \R_+} + h^{1/3}\|\psi\|_{\R \times \R_+}\big)
\end{equation}
holds for any $h \in (0,h_0)$, $\lambda \in D(\lambda_{1,\alpha}(0)h^{\frac23},Rh)$, $\psi \in \cHd$, and any $\mu$-subsolution $\Phi$.  
\end{proposition}
The proof is the object of the next section. It is completed in \S\ref{sec:proof_2D}.

\section{Elliptic estimate on the range of \texorpdfstring{$\Id - \Pi_{1,\alpha}$}{I-Pi\_1,alpha}}\label{sec:2D}
Recall the Riesz projector $\pi_1$ and the unitary operator $\mathcal{U}_\alpha(x)$ from \S \ref{sec:Pialpha}. For $x \in \R$, let  
$$\pi_{1,\alpha}(x) := \alpha(x)^{-2/3}\mathcal{U}_\alpha(x) \pi_1 \mathcal{U}_\alpha(x)^* = \frac{1}{2\pi i} \int_{\gamma} (\alpha(x)z - \mathscr{A}_\alpha(x))^{-1}\,dz.$$
Observe that by Proposition \ref{prop:u1} and rescaling, 
\begin{equation}
\label{eq:expression_pialpha}
\pi_{1,\alpha}(x) u = \langle u,\overline{u_{\alpha,x}}\rangle_{\R_+} u_{\alpha,x}
\end{equation}
where $u_{\alpha,x}$ is defined by \eqref{eq:def_ux}. 

Given $\eta > 0$, let
\begin{equation}
\label{eq:def_aeta}
a_\eta := (1 - \eta) |z_2| e^{i \pi/3},
\end{equation}
\begin{equation}
\label{eq:def_Geta}
G_\eta := \Big\{ z \in \mathbb{C} \,\mid\, \textup{Re} \left[ \big((\inf \alpha)^{2/3}a_\eta-z\big)e^{-i\frac{\pi}{4}}\right]\geq  0\Big\}\,;
\end{equation}
observe that $G_\eta$ is a translated and rotated half-plane, and that provided that the condition \eqref{eq:cond_alpha} in Assumption \ref{ass:alpha} holds, $G_\eta \cap \sigma(\mathscr{A}_\alpha(x)) = \{\lambda_{1,\alpha}(x)\}$ for all $x \in \R$ and $\eta$ small enough (see the sketch of the proof in \S\ref{sec:sketch}, especially Figure \ref{fig:half_plane_intro}). 

\begin{proposition}[Resolvent estimate for $\mathscr{A}_\alpha(x)$]
\label{prop:resAx}
Suppose that $\alpha$ satisfies Assumption \ref{ass:alpha}. Then, for any $\eta > 0$, there exists $C_\eta > 0$ such that for all $x \in \R$ and $z \in G_\eta \setminus \{\lambda_{1,\alpha}(x)\}$, 
\begin{equation}
\label{eq:res_Ax}
\|(\mathscr{A}_\alpha(x)-z)^{-1} (\Id - \pi_{1,\alpha}(x))\|_{\mathscr{L}(L^2(\R_+))} \leq \frac{C_\eta}{1 + |z|}.
\end{equation}
and 
\begin{equation}
\label{eq:res_Ax_D}
\|(\mathscr{A}_\alpha(x)-z)^{-1} (\Id - \pi_{1,\alpha}(x))\|_{\mathscr{L}(L^2(\R_+),\mathrm{D})} \leq C_\eta.
\end{equation}
\end{proposition}
Introducing the set 
$$F_\eta := \big\{z \in \mathbb{C} \,\mid\, \Re \left[e^{-i\frac\pi4}(a_\eta - z)\right] \geq 0\big\}\,,$$
Proposition \ref{prop:resAx} is a consequence of the following result and a scaling argument:
\begin{proposition}\label{prop:resA}
For any $\eta\in(0,1)$, there exist $C_\eta'>0$ such that for all $z \in F_\eta \setminus \{z_1 e^{i \frac{\pi}{3}}\}$, 
	\begin{align}
	\|(\mathscr{A}-z)^{-1}(\Id - \pi_1)\|_{\mathscr{L}(L^2(\R_+))} \leq \frac{C_\eta'}{1 + |z|}\,,
		\label{eq:estimA(I-Pi)}
	\end{align}
	\begin{equation}
		\|(\mathscr{A}-z)^{-1}(\Id - \pi_1)\|_{\mathscr{L}(L^2(\R_+),\mathrm{D})} \leq C_\eta'\,.
			\label{eq:estimA(I-Pi)D}
		\end{equation}
\end{proposition}
\begin{proof}[Proof of Proposition \ref{prop:resAx} from Proposition \ref{prop:resA}]
The definitions of $F_\eta$ and $G_\eta$ imply that for all $x \in \R$, 
$$z \in G_\eta \setminus \{\lambda_{1,\alpha}(x)\} \implies \alpha(x)^{-2/3} z \in F_\eta \setminus \{|z_1| e^{i\pi/3}\}.$$
Thus, given $z \in G_\eta \setminus \{\lambda_{1,\alpha}(x)\}$ and letting $\zeta := \alpha(x)^{-\frac23}z$, Proposition \ref{prop:resA} ensures that
$$\| (\mathscr{A} - \zeta)^{-1} (\Id - \pi_1)\|_{\mathscr{L}(L^2(\R_+))} \leq C_\eta',$$
and the $L^2$ estimate \eqref{eq:res_Ax} follows by writing (with $\mathcal{U}_\alpha(x)$ defined by \eqref{eq:defUx})
\begin{align*}
&\| (\mathscr{A}_\alpha(x) - z)^{-1} (\Id - \pi_{1,\alpha}(x))\|_{\mathscr{L}(L^2(\R_+),\mathrm{D})}\\
&\qquad\qquad = \big\| \mathcal{U}_\alpha(x) \big(\alpha(x)^{\frac23}\mathscr{A} - \alpha(x)^{\frac23} \zeta \big)^{-1}\alpha(x)^{-\frac23}(\Id - \pi_1) \mathcal{U}_\alpha(x)^*\big\|_{\mathscr{L}(L^2(\R_+),\mathrm{D})}\\
& \qquad\qquad \leq C\|(\mathscr{A} - \zeta)^{-1} (\Id - \pi_1)\|_{\mathscr{L}(L^2(\R_+))}
\end{align*}
since $\alpha^{-1}$ is bounded (by Assumption \ref{ass:alpha}) and $\mathcal{U}_\alpha(x)$ is unitary. The estimate \eqref{eq:res_Ax_D} is obtained similarly from \eqref{eq:estimA(I-Pi)D}, using that $\mathcal{U}_\alpha(x)$ maps $\mathrm{D}$ to itself continuously thanks to Assumption \ref{ass:alpha}. 
\end{proof}
To prove Proposition \ref{prop:resA}, we start by the following lemma. 
\begin{lemma}
\label{lem:estimees_faciles_A}
	For any $\theta_0 > 0$, there exists $C > 0$ such that the estimate
	\begin{equation}
	\label{eq:estim_AL2theta}
	\|(\mathscr{A}-z)^{-1}\|_{\mathscr{L}(L^2(\R_+))} \leq \frac{C}{|z|}
	\end{equation}
	holds for any complex number $z \neq 0$ in the sector $\frac{\pi}{2} + \theta_0 \leq \arg(z) \leq 2\pi - \theta_0$. 
\end{lemma}
\begin{proof}
For all $v \in \mathrm{D}$
\begin{equation}
\label{eq:theta_estim}
\Re (e^{-i\theta} \langle(\mathscr{A} -z)v,v\rangle) = \cos(\theta) \|D_y v\|^2 + \sin(\theta) \|\sqrt{y} v\|^2 - \textup{Re}(e^{-i\theta}z) \|v\|^2.
\end{equation}
In particular, since $\theta \in [0,\frac\pi2]$, 
$$\Re (e^{-i\theta} \langle(\mathscr{A}-z) v,v\rangle) \geq -\Re (e^{-i\theta}z) \|v\|^2.$$
Applying this to $v := (\mathscr{A}-z)^{-1}u$, using the Cauchy-Schwarz inequality and assuming $\Re(e^{-i\theta}z)< 0$, we deduce that
\begin{equation}
	\label{eq:estim_AL2}
	\|(\mathscr{A}-z)^{-1}\|_{\mathscr{L}(L^2(\R_+))} \leq \frac{1}{|\textup{Re}(e^{-i\theta} z)|}\,,
\end{equation}
Next, let $f : [\pi/2 + \theta_0,2\pi - \theta_0] \to [0,\frac{\pi}{2}]$ be such that $\cos(\varphi - f(\varphi)) < 0$ for all $\varphi \in [\pi/2+\theta_0,2\pi - \theta_0]$. For example, one can take 
$$f(\varphi) = \begin{cases}
0 & \textup{if } \varphi \in [\pi/2 + \theta,\pi]\\
\pi/4 & \textup{if } \varphi \in [\pi,3\pi/2)\\
\pi/2 & \textup{if } \varphi \in  [3\pi/2,2\pi - \theta_0].
\end{cases}$$
Then   
$$M := \min_{\varphi \in [\pi/2+\theta_0,2\pi - \theta_0]} \frac{1}{|\cos(\varphi-f(\varphi))|} > 0.$$
For $z = |z| e^{i\varphi}$ with $\varphi \in [\pi/2 + \theta_0,2\pi - \theta_0]$, we can apply \eqref{eq:estim_AL2} with $\theta := f(\varphi)$ to obtain
$$\|(\mathscr{A} - z)^{-1}\|_{\mathscr{L}(L^2(\R_+))} \leq \frac{M}{|z|}\,,$$
concluding the proof. 
\end{proof}

\begin{proof}[Proof of Proposition \ref{prop:resA}]
For $z \in F_\eta \setminus \{|z_1| e^{i\pi/3}\}$, let 
$$F(z) :=  \|(\mathscr{A}-z)^{-1}(\Id - \pi_1)\|_{\mathscr{L}(L^2(\R_+))}.$$ 
Given $\theta_0 \in (0, \frac{\pi}{4})$, Lemma \ref{lem:estimees_faciles_A} ensures that there exists $C_{\theta_0} > 0$ such that 
$$F(z) \leq \frac{C_{\theta_0}}{|z|}$$
for all $z \neq 0$ in the sector $\mathcal{S}(\theta_0) := \big\{ z\in \mathbb{C} \,\mid\,\frac{\pi}{2} + \theta_0 < \arg(z) <  2\pi - \theta_0\big\}$. This implies that for all $z \in \Omega := \mathcal{S}(\theta_0) \setminus \overline{D(0,1)}$,  
$$F(z) \leq \frac{2C_{\theta_0}}{1+ |z|}.$$
Therefore, it remains to show that there exists $C > 0$ such that 
$$\sup_{z \in K \setminus \{|z_1| e^{i\pi/3}\}} (1 + |z|)F(z) < +\infty$$ 
where 
$K := F_\eta \setminus \Omega$
is a compact set (see Figure \ref{fig:Feta}).
To this end, for $\delta > 0$ small enough, we write 
$$K = D(|z_1| e^{i\pi/3},\delta) \cup (K \setminus D(|z_1| e^{i\pi/3},\delta)).$$
It is immediate that $(1 + |z|)F(z)$ is bounded on $K \setminus D(|z_1| e^{i\pi/3},\delta)$ since this is a compact subset of $\rho(\mathscr{A})$. Finally, for $z \in D(|z_1| e^{i \frac\pi 3},\delta) \setminus \{|z_1|e^{i\frac\pi 3}\}$, we write
$$(\mathscr{A}-z)^{-1}(\Id - \Pi_1) = \frac{-1}{2\pi i} \int_{\mathscr{C}(|z_1| e^{i \frac\pi 3},2\delta)} (\zeta - z)^{-1} (\zeta - \mathscr{A})^{-1}\,d\zeta.$$
Thus, 
\begin{align*}
\|(\mathscr{A}-z)^{-1}(\Id - \Pi_1)\| 
&\leq \frac{1}{2\pi} \int_{\mathscr{C}(|z_1| e^{i \frac\pi 3},2\delta)} \delta^{-1} \|(\zeta - \mathscr{A})^{-1}\|d\zeta\\
& \leq 2\max_{\zeta \in \mathscr{C}(|z_1| e^{i \frac\pi 3},2\delta)} \|(\zeta - \mathscr{A})^{-1}\| < \infty
\end{align*}
since $\mathscr{C}(|z_1|e^{i\frac{\pi}{3}},2\delta)$ is again a compact subset of $\rho(\mathscr{A})$. This implies that $(1 + |z|) F(z)$ is also bounded on $D(|z_1| e^{i\frac\pi3},\delta) \setminus \{|z_1| e^{i\pi/3}\}$, completing the proof of the $L^2$ estimate \eqref{eq:estimA(I-Pi)}. 

In turn, the estimate \eqref{eq:estimA(I-Pi)D} in the $\mathrm{D}$ norm follows from the $L^2$ estimate \eqref{eq:estimA(I-Pi)} since
\begin{align*}
\|(\mathscr{A}-z)^{-1}(\Id - \pi_1)u\|_{\mathrm{D}}& = \|(\mathscr{A}-z)^{-1}(\Id - \pi_1)u\|_{L^2(\R_+)} + \|\mathscr{A}(\mathscr{A}-z)^{-1}(\Id - \pi_1)u\|_{L^2(\R_+)}\\
& \leq   (1 + |z|)\|(\mathscr{A}-z)^{-1}(\Id - \pi_1)u\|_{L^2(\R_+)} + \|u\|_{L^2(\R_+)}.\qedhere
\end{align*}
\end{proof}
\begin{figure}[htbp]
\centering
\begin{tikzpicture}[scale=0.65]
\draw (-3,0)--(6,0);
\draw (0,-3)--(0,6);
\coordinate (z1) at ({2.33811*cos(60)},{2.33811*sin(60)});
\coordinate (z2) at ({4.08795*cos(60)},{4.08795*sin(60)});
\coordinate (a1) at ($(z2)-(0.3,0.3) + 3.5*(-1,1)$);
\coordinate (a2) at ($(z2)-(0.3,0.3) - 5*(-1,1)$);
\coordinate (b1) at ({-0*cos(22.5)},{-0*sin(22.5)});
\coordinate (b2) at ({sin(10)},{-cos(10)});
\coordinate (B1) at ($(b1)-(b2)$);
\coordinate (B2) at ($(b1)-7*(b2)$);
\coordinate (c1) at ({-0*cos(80)},{-0*sin(80)});
\coordinate (c2) at ({sin(80)},{-cos(80)});
\coordinate (C1) at ($(c1)+7*(c2)$);
\coordinate (C2) at ($(c1)+(c2)$);
\coordinate (d1) at ({-1*cos(45)},{-1*sin(45)});
\coordinate (d2) at ({sin(45)},{-cos(45)});
\coordinate (D1) at ($(d1)+5*(d2)$);
\coordinate (D2) at ($(d1)-5*(d2)$);
\draw[purple] (B1)--(B2);
\draw[purple] (C1)--(C2);
\coordinate (aeta) at ($(z2)-(0.3,0.3)$);
\node[red] at ($(z1)$) {$+$};
\node[red] at ($(z1)+(0.4,-0.45)$) {$|z_1| e^{i \frac{\pi}{3}}$};
\node[red] at ($(z2)$) {$+$};
\node[red] at ($(z2)+(0.6,0.7)$) {$|z_2| e^{i \frac{\pi}{3}}$};
\draw[blue,dashed] (a1)--(a2);
\node[blue] at (-2,6) {$F_\eta$};
\node[purple] at (-2,2) {$\Omega$};
\node[purple] at (-1,7.5) {$\arg(z) = \frac{\pi}{2} + \theta_0$};
\node[purple] at (8.5,-0.5) {$\arg(z) = 2\pi - \theta_0$};
\node[blue,blue] at (aeta) {$+$};
\begin{scope}[on background layer]
      \fill[fill=orange!20!white,on background layer] (0,0)--($(b1)-6.17*(b2)$)--($(c1)+6.17*(c2)$)--cycle;
      \fill [orange!20!white] (0,0) circle (1);
   \end{scope}
	
\node[blue,blue] at ($(aeta)-(0.3,0.3)$) {$a_\eta$};
\node[orange] at (2.5,0.5) {$K$};
\draw [purple,domain=100:350] plot ({cos(\x)}, {sin(\x)});
\end{tikzpicture}
\caption{The regions of the complex plane involved in the proof of Proposition \ref{prop:resA}. The region $F_\eta$ is the half-plane lying to the south-west of the blue dashed line. The subset of $F_\eta$ which lies to the left and bottom of the solid purple line is the one denoted by $\Omega$ in this proof, and the orange shaded area is the compact set denoted by $K$.}
\label{fig:Feta}
\end{figure}
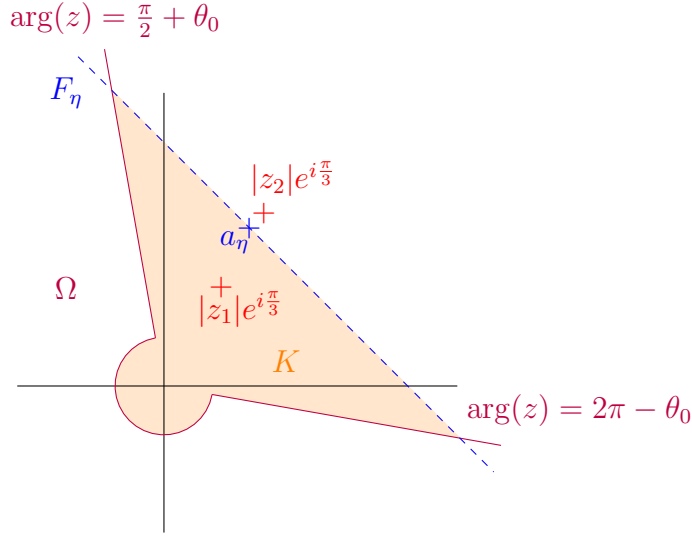

\subsection{Recap of semiclassical pseudodifferential operators with operator-valued symbols}
\label{sec:recap_pseudo}

This paragraph uses the results about semiclassical pseudodifferential calculus with operator-valued symbols from \cite[Chapter 2]{keraval2018formules}, \cite{keraval2024reduction} (in french), see also \cite[Appendix B]{gerard1991mathematical}. For the reader's convenience, we recall here the material needed for our purposes. The idea is to consider a family of operators $p(x,\xi)$ (considered as symbols), indexed by the variable $(x,\xi) \in \mathbb{R}^{2d}$, and to construct an operator $P = \textup{Op}_h^{w}(p)$ by a Weyl quantisation analogous to the one used with scalar-valued symbols.

\begin{definition}[{Ordered family of Hilbert spaces, \cite[Definition 2.1.1]{keraval2018formules}}]
\label{def:ordered}
A family $(\mathscr{H}_X)_{X \in \mathbb{R}^{2d}}$ of Hilbert spaces is an {\em ordered family of Hilbert spaces} if it satisfies the following properties:
\begin{enumerate}
\item For all $X \in \mathbb{R}^{2d}$, $\mathscr{H}_X = \mathscr{H}:= \mathscr{H}_{(0,0)}$.
\item There exists $C, N > 0$ such that for all $X,Y \in \R^{2d}$, $a \in \mathcal{A}$,
$$\|a\|_{\mathscr{H}_X} \leq C \langle X - Y \rangle ^N \|a\|_{\mathscr{H}_Y}.$$
\end{enumerate} 
\end{definition}
\begin{definition}[{Symbol class $S_\delta(\mathscr{H}_X,\mathscr{J}_X)$, \cite[Definition 2.1.2]{keraval2018formules}}]
Let $(\mathscr{H}_X)_{X \in \R^{2d}},(\mathscr{J}_X)_{\in \R^{2d}}$ be ordered families of Hilbert spaces. For $\delta \in [0,\frac12)$, the symbol class $S_\delta(\mathscr{H}_X,\mathscr{J}_X)$ is the set of families $(p_h)_{h \in (0,h_0)}$ of elements of $C^\infty(\R^{2d},\mathscr{L}(\mathscr{H},\mathscr{J}))$ such that for all $\beta \in \mathbb{N}^{2d}$, there exists $C_\beta > 0$ such that
$$\forall h \in (0,h_0)\,,\,\,\forall X \in \R^{2d}\,, \quad \|\partial^\beta p_h(X)\|_{\mathscr{L}(\mathscr{H}_X,\mathscr{J}_X)} \leq C_\alpha h^{-\delta|\beta|}.$$
In what follows, we denote
$$|(p_h)_{h \in (0,h_0)}|_{S_\delta(\mathscr{H}_X,\mathscr{J}_X),\beta} := \sup_{h \in (0,h_0)}\sup_{X \in \R^{2d}} h^{\delta |\beta|}\|\partial^\beta p_h(X)\|_{\mathscr{L}(\mathscr{H}_X,\mathscr{J}_X)},$$
or simply $|p_h|_{\delta,\beta}$ for short, when this will not lead to confusion.
\end{definition}
\begin{definition}[{Weyl quantization of operator-valued symbols, \cite[Definition 2.1.7]{keraval2018formules}}]
\label{def:Weyl}
Let $\mathscr{H}_X,\mathscr{J}_X$ be two ordered families of Hilbert spaces, let $\delta \in [0,\frac12)$ and let $(p_h)_{h \in (0,h_0)} \in S_\delta(\mathscr{H}_X,\mathscr{J}_X)$. Then the {\em Weyl quantization} $\textup{Op}^{w}_h(p_h)$ of $(p_h)_{h \in (0,h_0)}$ is defined by  
\begin{equation}
\label{eq:val_quantiz}
\big(\textup{Op}^{w}_h(p_h) u\big)(x) := \frac{1}{(2\pi h)^d} \iint_{\R^{2d}} e^{\frac{i}{h}(x-y)\cdot \xi} p_h\left(\tfrac{x+y}{2},\xi\right)u(y)\,dyd\xi\,, \quad x \in \mathbb{R}^d
\end{equation}
for $u \in \mathscr{S}(\R^d, \mathscr{H})$, where the integral is defined as a Bochner integral. 
\end{definition}
Just as in the scalar case, direct manipulations of the definition show that for $p(x,\xi) = p(x)$ and $q(x,\xi)= \xi^\alpha$  with $\alpha \in \mathbb{N}^d$,
\begin{equation}
\label{eq:Weyl_elem}
(\textup{Op}_h^w(p)u)(x) = p(x)u(x)\,, \quad \textup{Op}_h^w (q) u = h^{|\alpha|}D_\xi^\alpha u.
\end{equation}

We will use the following two main results which are analogous to the ones for semiclassical pseudodifferential operators with scalar-valued symbols:
\begin{theorem}[{Calder\'{o}n-Vaillancourt \cite[Theorem 2.1.16]{keraval2018formules}}]
\label{thm:CadlderonVaillancourt}
Suppose that for all $X \in \R^{2d}$, $\mathscr{H}_X \equiv \mathscr{H}$, $\mathscr{J}_X \equiv \mathscr{J}$. Let $\delta \in [0,\frac12)$, let $(p_h)_{h \in (0,h_0)} \in S_\delta(\mathscr{H},\mathscr{J})$. Then, the Weyl quantization  $\textup{Op}^{w}_h(p_h)$ in \eqref{eq:val_quantiz} extends to a unique linear continuous operator 
$$ \textup{Op}^{w}_h(p_h) : L^2(\R^d,\mathscr{H}) \to L^2(\R^d,\mathscr{J})$$
and there exists $C > 0$ and $M > 0$ such that for all $h \in (0,h_0)$ 
$$\|\textup{Op}^w_h(p_h)\|_{\mathscr{L}(L^2(\R^d,\mathscr{H}),L^2(\R^d,\mathscr{J}))} \leq C \sum_{\beta \leq Md} h^{|\beta|/2} \sup_{X \in \R^{2d}} \|\partial^\beta p_h(X)\|_{\mathscr{L}(\mathscr{H},\mathscr{J})}$$ 
\end{theorem}
\begin{theorem}[{Composition, \cite[Theorem 2.1.12]{keraval2018formules}}]
\label{thm:compo}
Let $(\mathscr{H}_X)_{X \in \R^{2d}}$, $(\mathscr{J}_X)_{X \in \R^{2d}}$, $(\mathscr{M}_X)_{X \in \R^{2d}}$ be sorted families of Hilbert spaces, let $\delta \in [0,\frac12)$ and let $(a_h)_{h \in (0,h_0)} \in S_\delta(\mathscr{H}_X,\mathscr{J}_X)$ and $(b_h)_{h \in (0,h_0)} \in S_\delta(\mathscr{J}_X,\mathscr{M}_X)$. Then there exists a unique symbol family $(c_h)_{h \in (0,h_0)} \in S_\delta(\mathscr{J}_X,\mathscr{M}_X)$ such that for all $h \in (0,h_0)$, 
$$\textup{Op}_h^w (a_h) \textup{Op}_h^w (b_h) = \textup{Op}_h^w (c_h).$$
For all $h$, we write 
$$a_h \# b_h := c_h.$$
The map $\# : S_{\delta}(\mathscr{H}_X,\mathscr{J}_X) \times S_\delta(\mathscr{J}_X,\mathscr{M}_X) \to S_\delta(\mathscr{H}_X,\mathscr{M}_X)$ is continuous, in the sense that for all $\beta \in \mathbb{N}^{2d}$, there exists $C_\beta, M_\beta > 0$ such that 
$$|a_h \# b_h|_{\delta,\beta} \leq C_\beta \sup_{|\beta'|\leq M_\beta} |a_h|_{\delta,\beta'} |b_h|_{\delta,\beta'}$$
For any $N \in \mathbb{N}^*$, $a_h \# b_h$ satisfies
$$a_h\# b_h = \sum_{k = 0}^N \frac{1}{k!} \left(\frac{ih}{2}\right)^{k} \left[\sigma(D_X,D_Y)^k(a_h(X)b_h(Y))\right]\big|_{Y = X} + h^{(N+1)(1-2\delta)}r_{N+1}$$
for some $r_{N+1} \in S_\delta(\mathscr{H}_X,\mathscr{M}_X)$, where, writing $X = (x,\xi)$ and $Y = (y,\eta)$, and denoting by $D_x,D_\xi,D_y,D_\eta$ the operators of differentiation corresponding to these variables, 
$$\sigma(D_X,D_Y)u(X,Y) = (D_\xi \cdot D_y - D_x \cdot D_\eta) u(X,Y).$$
\end{theorem}

For our purposes, only the following particular case of the composition formula will be required: 
$$a_h \# b_h = a_hb_h + \frac{h}{2i} \{a_h,b_h\} + h^{2(1-2\delta)} r_2\,,$$
where $r_2 \in S_\delta(\mathscr{H}_X,\mathscr{J}_X)$, and  
where $\{f,g\} = \partial_x f \partial_\xi g - \partial_\xi f \partial_x g$ is the Poisson bracket of $f$ and $g$.

\subsection{Construction of the approximate inverse}
\label{sec:construction_inverse}

We now use the previous theory to construct an approximate inverse of $\mathscr{L}_h^\Phi - \lambda$ by quantizing the inverse of its principal (operator-valued) symbol.

In view of the expression of $\mathscr{L}_h^\Phi$, its operator symbol is given by
$$p_h(x,\xi) = h^{2/3} \mathscr{A}_\alpha(x) + (\xi + i \Phi')^2 + i V(x).$$
(here and in what follows, we omit to signal the dependence of the symbol with respect to $\Phi$ in the notation). More precisely, we apply the setting of \S\ref{sec:recap_pseudo} with $d=1$ and we consider the following ordered families of Hilbert spaces.
\begin{definition}[The ordered families $\mathscr{H}_X$, $\mathscr{J}_X$, $\mathscr{M}_X$]
Let $(\mathscr{H}_X)_{X \in \R^2}$, $(\mathscr{J}_X)_{X \in \R^2}$, $(\mathscr{M}_X)_{X \in \R^2}$  be defined by $\mathscr{H} = \mathrm{D}$, $\mathscr{J} = \mathscr{M} =  L^2(\R_+)$, with the norms 
$$\|u\|_{\mathscr{H}_X} := \langle \xi \rangle^{2}\|u\|_{\mathrm{D}}\,, \quad  \|u\|_{\mathscr{J}_X}^2 := \|u\|_{L^2(\R_+)}^2\,, \quad \|u\|^2_{\mathscr{M}_X} = \langle \xi \rangle^{2} \|u\|^2_{L^2(\R_+)}.$$ 
\end{definition}
These families are ordered in the sense of Definition \ref{def:ordered}. In what follows, we abuse notation by writing $\mathrm{D}_{\langle \xi \rangle^2}$, $L^2(\R_+)$ and $L^2_{\langle \xi \rangle^2}(\R_+)$ instead of $\mathscr{H}_X$, $\mathscr{J}_X$ and $\mathscr{M}_X$, respectively. The following results follow immediately from the definitions and \eqref{eq:Weyl_elem}.

\begin{proposition}[Symbols of $\mathscr{L}_h^\phi$ and $\Pi_{1,\alpha}$]
\label{prop:symbols_Lh_Pi}
The symbol $(p_h)_{h \in (0,1)}$ belongs to $S_0(\mathrm{D}_{\langle \xi\rangle^2}, L^2(\R_+))$ and for any $\beta \in\mathbb{N}^2$, there exists $C_\beta > 0$ such that the estimates
\begin{equation}
\label{eq:class_estim_ph}
\forall \beta \in \mathbb{N}^2\,, \exists C_\beta\,:\quad |p_h|_{0,\beta} \leq C_\beta \max_{n \leq |\beta|} \|\partial_x^n \Phi'\|_\infty
\end{equation}
hold for any $\mu > 0$ and any $\mu$-subsolution $\Phi$. Moreover,
\begin{equation}
\label{eq:quantiz_Lphi}
\forall h > 0\,, \quad \mathscr{L}_h^\Phi = \textup{Op}^w_h(p_h).
\end{equation}
The symbol $x \mapsto \pi_{1,\alpha}(x)$ (independent of $\xi$ and $h$) belongs to $S_0(L^2(\R_+),L^2(\R_+))$ and
$$\forall h > 0\,, \quad \Pi_{1,\alpha} = \textup{Op}_h^w(\pi_{1,\alpha}).$$
\end{proposition}
\begin{proof}
Since $\mathscr{A}_\alpha(x) := \alpha(x) \mathcal{U}_\alpha(x) \mathscr{A} \mathcal{U}^*_\alpha(x)$ and $\pi_{1,\alpha}(x) = \alpha(x)^{-1} \mathcal{U}_\alpha(x) \pi_1 \mathcal{U}_\alpha^*(x)$, the class estimate for $p_h$ and $\pi_{1,\alpha}$ follow from the infinite differentiability with bounded derivatives of the maps $x \mapsto \alpha(x)$, $x \to \alpha^{-1}(x)$ from $\R$ to $\R$ and $x \mapsto \mathcal{U}_\alpha(x)$, $x \to \mathcal{U}^*_\alpha(x)$ from $L^2(\R_+)$ to $\mathrm{D}$ 
The fact that $\pi_{1,\alpha}$ belongs to $S_0(L^2(\R_+),L^2(\R_+))$ follows from the expression \eqref{eq:expression_pialpha} and the property \eqref{eq:Dxux}. The quantization formulas follow from \eqref{eq:Weyl_elem}.
\end{proof}

Observe that for all $\lambda \in \mathbb{C}$, 
\begin{equation}
\label{eq:z(lambda)}
p_h(x,\xi) - \lambda = h^{2/3}(\mathscr{A}_\alpha(x) - z_\lambda(x,\xi)) \quad \textup{where} \quad z_\lambda(x,\xi) := \frac{\lambda - [(\xi + i \Phi')^2 + iV]}{h^{2/3}}.
\end{equation}
\begin{lemma}
\label{lem:yotei}
There exists $\eta \in (0,1)$ and for all $\mu > 0$, there exists $C,c > 0$ such that the following holds. For all $R > 0$, there exists $h_0 > 0$ such that for all $h \in (0,h_0)$ and any $\mu$-subsolution $\Phi$,
\begin{equation}
\label{eq:implication_yotei}
\lambda \in D(\lambda_{1,\alpha}(0)h^{\frac23},Rh) \implies z_\lambda(x,\xi) \in G_\eta \textup{ for all } (x,\xi) \in \R^2,
\end{equation}
where $G_\eta$ is the region defined in \eqref{eq:def_Geta}.
\end{lemma}  
\begin{proof}
We start by observing that for all $(x,\xi) \in \R^2$, 
\begin{align}
\nonumber
\textup{Re} \left(e^{-i\frac{\pi}{4}} [(\xi + i \Phi')^2 + iV]\right) 
&= \cos(\pi/4) \left(\xi^2 - \Phi'^2 + 2\xi \Phi' + V\right)\\
& = \cos(\pi/4) \left((\xi + \Phi')^2 + V - 2 \Phi'^2\right) \geq 0
\end{align}
since $\Phi$ is a $\mu$-subsolution.

Recalling the definition of $a_\eta$ from \eqref{eq:def_aeta}, for $\lambda \in D(h^{2/3}\lambda_{1,\alpha}(0),Rh)$ and $(x,\xi) \in \R^2$, we deduce that
\begin{align*}
\textup{Re}(e^{-i\pi/4}z_\lambda(x,\xi)) 
&\leq \Re \left[e^{-i\pi/4} \lambda_{1,\alpha}(0)\right] + Rh^{1/3} \\
&\leq \cos(\pi/12) \left(|z_1| (\sup \alpha)^{2/3} + \frac{Rh^{1/3}}{\cos(\pi/12)}\right)\\
&\leq \cos(\pi/12) (1-\eta)|z_2| (\inf \alpha)^{2/3} = \textup{Re}(e^{-i\pi/4}a_\eta)
\end{align*} 
for $h$ and $\eta$ small enough, since $|z_1| \sup(\alpha)^{2/3} < |z_2| \inf(\alpha)^{2/3}$ by Assumption \ref{ass:alpha}. 
\end{proof}

\begin{proposition}
\label{prop:def_rlambda}
Given $\mu \in (0,1)$, there exists $C,c > 0$ such that the following holds. For all $R > 0$, there exists $h_0 > 0$ such that for any $\mu$-subsolution $\Phi$, $h \in (0,h_0)$, $(x,\xi) \in \R^2$ and $\lambda \in D(h^{\frac23}\lambda_{1,\alpha}(0),Rh)$, there exists a unique $r_\lambda(x,\xi) \in \mathscr{L}(L^2(\R_+))$ such that the following properties hold:
\begin{itemize}
\item[(i)]$r_\lambda(x,\xi) (p(x,\xi) - \lambda) =(p(x,\xi) - \lambda) r_\lambda(x,\xi)  = \Id - \pi_{1,\alpha}(x)$.
\item[(ii)] The map $\lambda \to r_\lambda(x,\xi)$ is continuous.
\item[(iii)] $r_\lambda(x,\xi) (\Id - \pi_{1,\alpha}(x)) =  (\Id - \pi_{1,\alpha}(x)) r_\lambda(x,\xi)= r_\lambda(x,\xi)$. 
\end{itemize}
Moreover, for every $\lambda \in D(h^{2/3}\lambda_{1,\alpha}(0),Rh)$, $r_\lambda \in C^\infty(\R^2;\mathscr{L}(L^2(\R_+)))$ and 
\begin{equation}
\label{eq:eric_prisonnier}
\langle \xi \rangle^2\|r_\lambda(x,\xi)\|_{\mathscr{L}(L^2(\R_+))} \leq Ch^{-2/3}\,, \quad \|y\cdot r_\lambda(x,\xi)\|_{\mathscr{L}(L^2(\R_+))} \leq C h^{-2/3}
\end{equation}

\end{proposition}
\begin{proof}
Given $x \in \R$, we denote $E(x):= (\Id- \pi_{1,\alpha}(x))(\mathrm{D})$ endowed with the norm $\|\cdot\|_{\mathrm{D}}$. 
Since $\mathscr{A}_\alpha(x)$ commutes with $\pi_{1,\alpha}(x)$ for all $x \in \R$, we may define for each $x \in \R$ and $z \in \mathbb{C}$ an operator $L_x(z) \in \mathscr{L}(E(x))$ by
$$L_x(z) u := h^{2/3}(\mathscr{A}_\alpha(x) - z).$$
The application $z \mapsto L_x(z)$ defines a function $L_x : \mathbb{C} \to \mathscr{L}(E(x))$ which is holomorphic. Moreover, with $z_\lambda(x,\xi)$ as in \eqref{eq:z(lambda)}, 
\begin{equation}
\label{eq:lienLph}
L_x(z_\lambda(x,\xi)) = (p_h(x,\xi) - \lambda)|_{E(x)}.
\end{equation}

Next, let
$$\Omega := D(\lambda_{1,\alpha}(0),Rh^{1/3}) - \frac{[(\xi + i \Phi')^2 + iV]}{h^{2/3}} \subset \mathbb{C}.$$
By Lemma \ref{lem:yotei}, for $h$ small enough, $\Omega \subset G_\eta$,  and in particular $\Omega \cap \sigma(\mathscr{A}_\alpha(x)) \subset \{\lambda_{1,\alpha}(x)\}$.
Hence, for $z \in \Omega \setminus \{\lambda_{1,\alpha}(x)\}$, we can define another map $F_x(z) \in \mathscr{L}(E(x))$ by
$$F_x(z):= (\mathscr{A}_\alpha(x) - z)^{-1} (\Id - \Pi_{1,\alpha}(x)).$$
For each $x$, the map $z \mapsto F_x(z)$ is holomorphic on $\Omega \setminus \{\lambda_{1,\alpha}(x)\}$ and bounded uniformly for $z$ in a neighborhood of $\lambda_{1,\alpha}(x)$ thanks Proposition \ref{prop:resAx} and the fact that $\Omega \subset G_\eta$. Therefore, it admits a holomorphic extension to the whole $\Omega$ that we denote by $\widetilde{F}_x$. Moreover the combination of Proposition \ref{prop:resAx} and Lemma \ref{lem:yotei} implies that
\begin{equation}
\label{eq:estim_prolongee}
(1 + |z_\lambda(x,\xi)|)\|\widetilde{F}_x(z_\lambda(x,\xi))\|_{\mathscr{L}(L^2(\R_+))} \leq C
\end{equation}
\begin{equation}
\label{eq:estim_prolongeeD}
\|\widetilde{F}_x(z_\lambda(x,\xi))\|_{\mathscr{L}(L^2(\R_+),\mathrm{D})} \leq C
\end{equation}
Using the fact that $\Phi$ is a $\mu$-subsolution, and Assumption \ref{ass:V}, we notice that $|z_\lambda(x,\xi)| \geq ch^{-2/3}|\xi|^2 - Ch^{-2/3}$ (since $|\Phi'| \leq \frac{1-\mu}{2}\sqrt{V}$); thus \eqref{eq:estim_prolongee} implies  
\begin{equation}
\label{eq:estim_prolongee_bis}
\langle \xi \rangle^2 \|\widetilde{F}_x(z_\lambda(x,\xi))\|_{\mathscr{L}(L^2(\R_+))} \leq C.
\end{equation}
On the other hand, using the elliptic estimate of Proposition \ref{prop:ellip_A_yD2y}, \eqref{eq:estim_prolongeeD} implies that
\begin{equation}
\label{eq:estim_prolongeeDbis}
\|y \cdot \widetilde{F}_x(z_\lambda(x,\xi))\|_{\mathscr{L}(L^2(\R_+))} \leq C
\end{equation}

For all $z \in \Omega \setminus \{\lambda_{1,\alpha}(x)\}$, $F_x(z) L_x(z) = L_x(z) F_x(z) = h^{2/3}\Id_{E(x)}$ and thus by continuity,
\begin{equation}
\label{eq:FxLxLxFx}
\forall z \in \Omega\,, \quad \widetilde{F}_x(z) L_x(z) = \widetilde{F}_x(z) L_x(z) = h^{2/3}\Id_{E(x)}
\end{equation}
Hence, if we define $r_\lambda(x,\xi) \in \mathscr{L}(L^2(\R_+),\mathrm{D})$ by
$$r_\lambda(x,\xi) := h^{-2/3}\widetilde{F}_x(z_\lambda(x,\xi))(\Id - \pi_{1,\alpha}(x))\,,$$
then by \eqref{eq:FxLxLxFx} and \eqref{eq:lienLph}, 
$$r_\lambda(x,\xi) (p_h(x,\xi) - \lambda) =  (p_h(x,\xi) - \lambda) r_\lambda(x,\xi) = \Id - \pi_{1,\alpha}(x)$$
which proves (i). The continuity (ii) is immediate. The property (iii) follows from the definition of $r_\lambda(x,\xi)$ if $\lambda \neq \lambda_{1,\alpha}(x)$ and by continuity for $\lambda = \lambda_{1,\alpha}(x)$. The uniqueness is immediate for $\lambda \neq h^{2/3}\lambda_{1,\alpha}(x) - [(\xi + i \phi')^2 + iV]$ since in this case the operator $p(x,\xi) - \lambda$ is invertible, and follows for $\lambda$ in the whole disk $D(\lambda_{1,\alpha}(0)h^{\frac23},Rh)$ by continuity. Since $\widetilde{F}_\lambda$ is holomorphic on $\Omega$ and $(x,\xi) \mapsto z_\lambda(x,\xi)$ is infinitely differentiable, we deduce by composition that $(x,\xi) \mapsto r_\lambda(x,\xi)$ is infinitely differentiable. Finally, the estimates in \eqref{eq:eric_prisonnier} follow from the boundedness of $\Id - \pi_{1,\alpha}(x)$, \eqref{eq:estim_prolongee_bis} and \eqref{eq:estim_prolongeeDbis}. 
\end{proof}

\subsection{Class estimates for \texorpdfstring{$r_\lambda$}{r\_lambda}}
\label{sec:class_r}
To exploit the symbol $r_\lambda$ constructed above, we need to ensure that it belongs to a suitable symbol class. To this end, we prove the following result.
\begin{proposition}
\label{prop:rlambda_class}
For all $R > 0$ and $\mu > 0$, there exists $h_0 > 0$ such that for all $h \in (0,h_0)$ and all $\lambda \in D(\lambda_{1,\alpha}(0)h^{\frac23},Rh)$, the symbol $(h^{2/3}r_\lambda)_{h \in (0,h_0)}$ belongs to $S_{\frac13}(L^2(\R_+),L_{\langle \xi \rangle^2}^2(\R_+))$. More precisely, for $\mu > 0$ and $\beta \in \mathbb{N}^2$, there exists $C > 0$ such that the estimate
$$|h^{2/3}r_\lambda|_{\frac13,\beta} \leq C \max_{n \leq |\beta|} \|\partial^n_x \Phi'\|_\infty$$ 
holds for all $\lambda \in D(\lambda_{1,\alpha}(0)h^{2/3},Rh)$ and $\Phi$ any $\mu$-subsolution.
\end{proposition}
Let us first summarize our method. The starting point is the following formulas for the derivatives of $r_\lambda$.
\begin{lemma}[First derivatives of $r_\lambda$]
\label{prop:dxdxir}
Let $\lambda \in D(\lambda_{1,\alpha}(0)h^{2/3},Rh)$ and let $r_\lambda$ be defined by Proposition \ref{prop:def_rlambda}. Then 
\begin{align}
\label{eq:dxr}
\partial_x r_\lambda & = -r_\lambda (\partial_x p) r_\lambda - r_\lambda (\partial_x \pi_{1,\alpha}) - (\partial_x \pi_{1,\alpha}) r_\lambda.\\
\label{eq:dxir}
\partial_\xi r_\lambda & = -r_\lambda(\partial_\xi p) r_\lambda.
\end{align}
\end{lemma} 
\begin{proof}
This is obtained by differentiating the relation (i) of Proposition \ref{prop:def_rlambda} and using the relation (iii) of this proposition. 
\end{proof}
The key point is that although $\|r_\lambda\|_{L^2 \to L^2_{\langle \xi \rangle^2}}$ is $O(h^{-2/3})$, the terms $\|(\partial_xp) r_\lambda\|_{\mathscr{L}(L^2)}$ and $\|(\partial_\xi p)r_\lambda\|_{\mathscr{L}(L^2)}$, and more generally, all terms of the form $\|(\partial^\beta p)r_\lambda\|_{\mathscr{L}(L^2)}$ for $|\beta|\geq 1$, satisfy the better bounds $O(h^{-|\beta|/3})$. Finally, differentiating more times $r_\lambda$ leads to linear combinations of terms of the form $r_\lambda (\partial^{\beta_1} p) r_\lambda (\partial^{\beta_2} p) r_\lambda \ldots (\partial^{\beta_J} p) r_\lambda$, or ``better terms". 

To formalize these ideas, we start by capturing precisely the structure of the $\beta$-th derivatives of $r_\lambda$. 

\subsubsection*{Structure of the derivatives of $r_\lambda$}
The idea is that $\partial^\beta r_\lambda$ is an $N$-term in the sense of Definition \ref{def:Nterm} below, with $N$-terms constructed from $N$-atoms that we define now.

\begin{definition}[$N$-atom]
\label{def:Natom}
Given $N \in \mathbb{N}$ and $h_0 > 0$, we say that a family $(a_N(h))_{h \in (0,h_0)}$ of elements of $C^\infty(\R^2; \mathscr{L}(L^2(\R_+))$ is an $N$-atom if either
\begin{itemize}
\item[(i)] $N = 0$ and $a_N(h) = f_h$ for all $h \in (0,h_0)$, where $(f_h)_{h \in (0,h_0)} \in S_0(L^2(\R_+),L^2(\R_+))$ is independent of $\Phi$.
\item[(ii)] $N \geq 1$ and for all $h \in (0,h_0)$, 
$$a_N(h) = f_{1,h}\cdot  (\partial^\beta p) \cdot f_{2,h} \cdot r_\lambda \cdot f_{3,h}$$ where $\beta \in \mathbb{N}^2$ 
satisfies $|\beta| = N$, and where $(f_{i,h})_{h \in (0,h_0)} \in S_0(L^2(\R_+),L^2(\R_+))$ are independent of $\Phi$, for ${i = 1,2,3}$. 
\item[(iii)] $N \geq 1$ and there exists $i,j\in \{1,\ldots,N-1\}$ satisfying $i + j \leq N$ and 
$$a_N(h) = a_i(h) \cdot  a_j(h)$$
where $(a_i(h))_{h \in (0,h_0)}$ is an $i$-atom and $(a_j(h))_{h \in (0,h_0)}$ is a $j$-atom. 
\end{itemize}
\end{definition}
In what follows, when it will not lead to confusion, we omit the dependence in $h$ from the notation. Observe that an $N$-atom $a_N$ is also an $M$-atom for any $M \geq N$ (writing it as $1 \cdot a_N$ in  case (iii), using that $1$ is a $0$-atom by case (i)). 
\begin{definition}[$N$-term]
\label{def:Nterm}
Given $N \in \mathbb{N}$ and $h_0 > 0$, we say that a family $(t_N(h))_{h \in (0,h_0)}$ is an $N$-term if it is of the form 
$$t_N(h) = \sum_{j = 1}^J f_{j,h} \cdot r_\lambda \cdot a_{N,j}(h)$$
for all $h \in (0,h_0)$, where $J \in \mathbb{N}$, $(f_{j,h})_{h \in (0,h_0)} \in S_0(L^2(\R_+),L^2(\R_+))$ are independent of $\Phi$ and $a_{N,j}$ are $N$-atoms. 
\end{definition}
Similarly, we may omit the $h$-dependence from the notation of $N$-terms. Observe that if $t_N$ is an $N$-term, it is also an $M$-term for every $M \geq N$. 

\begin{proposition}
\label{prop:fdxdxiatom}
If $a_N$ is an $N$-atom and $f \in S_0(L^2(\R_+),L^2(\R_+))$, then $f \cdot a_N$ is an $N$-atom. Moreover, $\partial_x a_N$ and $\partial_\xi a_N$ are linear combinations of a finite number of $(N+1)$-atoms.
\end{proposition}
\begin{proof}
The first statement is easily shown by definition and induction on $N$. For the second statement, we only consider the case of the $x$-derivative (the other case is similar, but simpler). We proceed by induction on $N$ and consider the three cases (i)-(iii) of Definition \ref{def:Natom} separately. 
\begin{enumerate}[1.]
\item In case (i), the result is obvious. Hence, suppose that the result holds for all $j$-atoms with $j \leq N-1$ and let $a_{N}$ be an $N$-atom.
\item In case (ii), the only difficulty is to show that
$$b_{N+1} := f_1\cdot \partial^\beta p \cdot f_2 \cdot (\partial_x r_\lambda) \cdot f_3$$
is a linear combination of $(N+1)$-atoms. Using the expression of the derivatives of $r_\lambda$ from Proposition \ref{prop:dxdxir}, this term reads
\begin{align*}
b_N &= f_1\cdot \partial^\beta p \cdot f_2 \cdot \big(-r_\lambda (\partial_x p)r_\lambda - \partial_x \pi_{1,\alpha} r_\lambda - r_\lambda \partial_x \pi_{1,\alpha} \big) \cdot f_3 \\
& = c_N \cdot c_1 - f_1 \cdot \partial^\beta p \cdot \widetilde{f}_2 \cdot r_\lambda \cdot f_3 - f_1 \cdot \partial^\beta p \cdot f_2 \cdot r_\lambda \cdot \widetilde{f}_3 
\end{align*}
where $c_N = -f \cdot (\partial^\beta p) \cdot g \cdot r_\lambda$, $c_1 = \partial_x p\cdot r_\lambda$, $\widetilde{f}_2 = f_2 \cdot \partial_x \pi_{1,\alpha}$ and $\widetilde{f}_3 = \partial_x \pi_{1,\alpha}\cdot h$. Thus $b_{N+1}$ is indeed a linear combination of $(N+1)$-atoms. 
\item In case (iii), the result follows by writing $\partial_x a_N = \partial_x a_i \cdot a_j + a_i \cdot \partial_x a_j$ and applying the induction hypothesis. 
\end{enumerate} 
This concludes the proof.
\end{proof}
\begin{corollary}
\label{cor:dxdxiNterm}
If $t_N$ is an $N$-term, then $\partial_x t_N$ and $\partial_\xi t_N$ are $(N+1)$-terms.   
\end{corollary}
\begin{proof}
We only treat the $x$-derivative (the case of the $\xi$-derivative is similar and simpler). Given an $N$-term $t_N$, we write 
\begin{align*}
\partial_x t_N &= \sum_{j = 1}^J (\partial_x f_j) \cdot r_\lambda \cdot a_{N,j} - \sum_{j = 1}^J f_j r_\lambda \overbrace{ \underbrace{[(\partial_x p) \cdot r_\lambda]}_{\textup{$1$-atom}}\cdot \underbrace{a_{N,j}}_{\textup{$N$-atom}}}^{(N+1)\textup{-atom}} \\
& \qquad - \sum_{j = 1}^J (f_j \partial_x \pi_{1,\alpha})\cdot r_\lambda \cdot  a_{N,j} - \sum_{j = 1}^J f_j \cdot r_\lambda \cdot \underbrace{[(\partial_x\pi_{1,\alpha})a_{N,j}]}_{N\textup{-atom}} + \sum_{j = 1}^J f_j \cdot r_\lambda \cdot \underbrace{(\partial_x a_{N,j})}_{\sum(N+1)\textup{-atoms}},
\end{align*}
and the conclusion follows by Proposition \ref{prop:fdxdxiatom}.
\end{proof} 
\begin{corollary}
\label{cor:rlambda_Nterm}
For any $\beta \in \mathbb{N}^2$, $\partial^\beta r_\lambda$ is a $|\beta|$-term.
\end{corollary}
\begin{proof}
This follows by an immediate induction from Corollary \ref{cor:dxdxiNterm}, since $r_\lambda$ is a $0$-term.
\end{proof}

\subsubsection*{Estimates of $N$-atoms and proof of Proposition \ref{prop:rlambda_class}}
To proceed, we now estimate the norms of $N$-atoms.
\begin{proposition}
\label{prop:estim_atom}
If $(a_N(h))_{h \in (0,h_0)}$ is an $N$-atom, then there exists $C > 0$ such that for all $h \in (0,h_0)$ and for any $\mu$-subsolution $\Phi$,  
$$\|a_N(h)\|_{\mathscr{L}(L^2(\R_+))} \leq C \max_{n \leq N} \|\partial_x^{n} \Phi'\|_\infty h^{-N/3}.$$
\end{proposition} 
The proof relies on the following lemma
\begin{lemma}
\label{lem:h13}
For all $R > 0$ and $\mu > 0$, there exists $h_0 > 0$ and $C > 0$ such that the estimates
\begin{equation}
\label{eq:xi+iphi'}
\|(\xi + i\Phi')r_\lambda(x,\xi)\|_{\mathscr{L}(L^2)} \leq C h^{-1/3}\,
\end{equation}
\begin{equation}
\label{eq:V'}
\|V'(x) r_\lambda(x,\xi)\|_{\mathscr{L}(L^2)} \leq Ch^{-1/3}
\end{equation}
holds for all $h \in (0,h_0)$, $(x,\xi) \in \R^2$, $\lambda \in D(h^{2/3}\lambda_{1,\alpha}(0),Rh)$ and any $\mu$-subsolution $\Phi$. 
\end{lemma}
\begin{proof}[Proof of Proposition \ref{prop:estim_atom}]
We claim that for $\beta \in \mathbb{N}^2$ with $\beta \neq 0$ and $f \in S_0(L^2(\R_+),L^2(\R_+))$, 
\begin{equation}
\label{eq:estim_betap}
\|(\partial^\beta p) \cdot f \cdot r_\lambda\|_{\mathscr{L}(L^2(\R_+))}\leq C \max_{n \leq |\beta|}\|\partial^{n}_x\Phi'\|h^{-|\beta|/3}.
\end{equation}
If this is true, then the result of the lemma is easily obtained by induction on $N$. To show \eqref{eq:estim_betap}, we write $\partial^\beta = \partial^{\beta_1}_x \partial^{\beta_2}_\xi$ and consider the following cases separately
\begin{enumerate}[(i)]
\item If $\beta_2 \geq 2$, then $|\beta|\geq 2$, $\|\partial^\beta p\|_{\mathscr{L}(L^2(\R_+))} \leq C \max_{n \leq |\beta|}\|\partial_x^n \Phi'\|_\infty$, and thus 
\begin{align*}
\|(\partial^\beta p) \cdot f \cdot r_\lambda \| &\leq C  \max_{n \leq |\beta|}\|\partial_x^n \Phi'\|_\infty \|r_\lambda\|_{\mathscr{L}(L^2(\R_+))}\\
& \leq C \max_{n \leq |\beta|}\|\partial_x^n \Phi'\|_\infty h^{-\frac23}\leq C  \max_{n \leq |\beta|}\|\partial_x^n \Phi'\|_\infty h^{-|\beta|/3}.
\end{align*}
by the estimate \eqref{eq:eric_prisonnier} of Proposition \ref{prop:def_rlambda}.
\item If $\beta_2 = 1$ and $\beta_1 = 0$, then $\partial^\beta p(x,\xi) = 2(\xi + i\Phi')$, which is a scalar, and thus
$$\|\partial^\beta p \cdot f \cdot r_\lambda\| = \|f \cdot [(\xi + i\Phi') r_\lambda]\| \leq Ch^{-\frac13} = Ch^{-\frac{|\beta|}{3}}$$
by Lemma \ref{lem:h13}. 
\item If $\beta_2 = 1$ and $\beta_1 \geq 1$, then $|\beta|\geq 2$,  $\partial^\beta p  = \partial^{\beta_1}_x \Phi'$, and thus we conclude as in case (i). 
\item If $\beta_2 = 0$ and $\beta_1= 1$, then $\partial^\beta p = h^{\frac23}\partial_x \mathscr{A}_\alpha + 2(\xi + i\Phi') \Phi'' + i V'$ and thus 
$$\|\partial^\beta p \cdot f \cdot r_\lambda\| \leq C\|f\|_{\mathscr{L}(L^2(\R_+)} \left( h^{2/3} \|r_\lambda\|_{\mathscr{L}(L^2(\R_+))} + C\|\Phi''\|_\infty \|(\xi + i \phi') r_\lambda\| + \|V' r_\lambda\|\right)$$
and we conclude using Lemma \ref{lem:h13} and the estimate \eqref{eq:eric_prisonnier}. 
\item In the remaining case where $\beta_2= 0$ and $\beta _1 \geq 2$, We can write 
$$\partial^\beta p = (\xi + i\Phi')(\partial_x^{\beta_1}\Phi') + g(x)$$
where $\|g\|_{\mathscr{L}(L^2(\R_+))} \leq C \max_{n \leq |\beta|} \|\partial_x^n \Phi'\|_{\infty}$, and thus 
\begin{align*}
\|\partial^\beta p \cdot f \cdot r_\lambda\| 
&\leq C \max_{n \leq |\beta|} \|\partial_x^n \Phi'\|_\infty \left( \|r_\lambda\|_{\mathscr{L}(L^2(\R_+))} + \|(\xi + i\Phi')r_\lambda\|_{\mathscr{L}(L^2(\R_+))}\right) \\
&\leq C \max_{n \leq |\beta|} \|\partial_x^n \Phi'\|_\infty h^{-2/3} \leq C \max_{n \leq |\beta|} \|\partial_x^n \Phi'\|_\infty h^{-|\beta|/3},
\end{align*}
again by Lemma \ref{lem:h13} and the estimate \eqref{eq:eric_prisonnier}. \qedhere
\end{enumerate}
\end{proof} 

\begin{proof}[Proof of Lemma \ref{lem:h13}]
~
We start by showing the estimate \eqref{eq:xi+iphi'}. 
\begin{enumerate}[1.]
\item When $\xi^2 + V \leq C h^{2/3}$, the estimate follows by writing
$$\|(\xi + i\Phi')r_\lambda(x,\xi)\|_{L^2 \to L^2}= \sqrt{\xi^2 + \Phi'^2} \|r_\lambda\|_{L^2 \to L^2},$$
using that $\|r_\lambda\| \leq Ch^{-2/3}$ by Proposition \ref{prop:def_rlambda} and the fact that $\Phi'^2 \leq \frac{(1 - \mu)}{2}V$. 
\item For $\xi^2 + V \geq Ch^{2/3}$, by adding the estimates 
\[\Re\left[(p(x,\xi)-\lambda)e^{-i\frac{\pi}{4}}\right]\geq cV(x)-C_0h^{\frac23}\,,\quad\Re (p(x,\xi)-\lambda)\geq \xi^2-\Phi'^2-C_0h^{\frac23}.\]
one obtains 
\begin{equation}
\label{eq:arnaque}
(\xi^2 + V)\|r_\lambda u\| \leq C \|u\|.
\end{equation}
for $h$ small enough, using that $\Phi$ is a $\mu$-subsolution. ¨The result follows by writing 
$$\|(\xi + i \phi')r_\lambda\| \leq C \frac{(\xi^2 + V)\|r_\lambda\|}{(\xi^2 + V)^{1/2}}.$$
\end{enumerate}
We now turn to the estimate \eqref{eq:V'}.
\begin{enumerate}[1.]
\item If $|x| \leq Ch^{1/3}$, then $|V'| \leq Ch^{1/3}$ since $V'(x) \sim V''(0)x$ as $x \to 0$. Thus the result follows immediately from the estimate $\|r_\lambda\|_{\mathscr{L}(L^2(\R_+))} \leq Ch^{-2/3}$ of Proposition \ref{prop:def_rlambda} in this case.
\item If $|x| \geq Ch^{1/3}$ then, since  (by Assumption \ref{ass:V})  (i) $V(x)$ is quadratic at $x = 0$, (ii) $V(x) > 0$ for $x \neq 0$ and (iii) $V$ is bounded below at infinity, there exists $c$ such that $V(x) \geq ch^{2/3}$ for all $x \in \R$. Moreover, one has $|V'(x)|/V(x) \leq Ch^{-1/3}$ for all $x \in \R$ (since $V'/V \sim \frac12x^{-1}$ as $x \to 0$). Thus, using again the estimate \eqref{eq:arnaque} above,
$$\|V'r_\lambda\| \leq \frac{|V'(x)|}{V(x)} V(x)\|r_\lambda\| \leq h^{-1/3}(\xi^2 + V) \|r_\lambda\| \leq Ch^{-1/3}\,,$$
concluding the proof. \qedhere 
\end{enumerate}
\end{proof}

\begin{proof}[Proof of Proposition \ref{prop:rlambda_class}]
By Corollary \ref{cor:rlambda_Nterm}, $\partial^\beta r_\lambda$ can be written in the form 
$$\partial^\beta r_\lambda = \sum_{j = 1}^J f_{j} \cdot r_\lambda \cdot a_{N,j}$$
where $f_j \in S_0(L^2(\R_+),L^2(\R_+))$ and $a_{N,j}$ are $N$-atoms with $N = |\beta|$ in the sense of Definition \ref{def:Natom}. Therefore by Proposition \ref{prop:estim_atom} and the first estimate in \eqref{eq:eric_prisonnier} in Proposition \ref{prop:def_rlambda}, 
\begin{align*}
\|r_\lambda\|_{\mathscr{L}(L^2(\R_+),L^2_{\langle \xi \rangle^2}(\R_+)} 
&\leq \sum_{j = 1}^J \|f_j\|_{\mathscr{L}(L^2_{\langle \xi\rangle^2}(\R_+))} \|r_\lambda\|_{\mathscr{L}(L^2(\R_+),L^2_{\langle \xi \rangle^2}(\R_+))} \|a_{N,j}\|_{\mathscr{L}(L^2(\R_+))}\\
&\leq C \max_{n \leq N} \|\partial_x^n \Phi'\|_\infty h^{-2/3} h^{-N/3},
\end{align*}
since $\|f_j\|_{\mathscr{L}(L^2_{\langle \xi\rangle^2}(\R_+))} = \|f_j\|_{\mathscr{L}(L^2(\R_+))}$. This concludes the proof.
\end{proof}

\subsection{Proof of Proposition \ref{prop:ellip2}}
\label{sec:proof_2D}

By the quantization formulas of Proposition \ref{prop:symbols_Lh_Pi} the composition formula, the continuity of $\#$ and the class estimates of Proposition \ref{prop:rlambda_class}, we can write 
\begin{equation}
\label{eq:composition}
(\Id - \Pi_{1,\alpha}) = h^{-2/3}R_\lambda (\mathscr{L}^\Phi_h - \lambda) - ih^{1/3}K_1 + h^{2/3}K_2
\end{equation}
where
$$R_\lambda := \textup{Op}_h^w(h^{2/3}r_\lambda)\,, \quad K_1 = \textup{Op}_h^w(\{(h^{2/3}r_\lambda),p_h\})\,, \quad K_2 = \textup{Op}_h^w(\mu_h)$$
with 
\begin{equation}
\label{eq:class_mu}
(\mu_h)_{(0,h_0)} \in S_{1/3}(\mathrm{D}_{\langle \xi\rangle^2},L_{\langle \xi \rangle^2}^2(\R_+)) = S_{1/3}(\mathrm{D},L^2(\R_+))
\end{equation}
and where for any $\beta \in\mathbb{N}^2$, there exists $C_\beta$ and $M_\beta$ such that
$$|\mu_h|_{\frac13,\beta} \leq C_\beta \max_{n \leq M_\beta} \|\partial_x^n \Phi'\|_\infty.$$
The last equality in \eqref{eq:class_mu} eliminates the dependence in $\xi$ from the spaces, and thus allows us to apply the Calder\'{o}n-Vaillancourt theorem (Theorem \ref{thm:CadlderonVaillancourt}) to $\textup{Op}_h^w(\mu_h)$. Namely, there exists $N > 0$ such that 
$$\|K_2 u\|_{\mathcal{H}} \leq C\max_{n \leq N} \|\partial_x^N \Phi'\| \|u\|_{L^2(\R;\mathrm{D})},$$
where $\|u\|^2_{L^2(\R,\mathrm{D})} := \int_{\R} \|u(x,\cdot)\|_{\mathrm{D}}^2\,dx$. 
Moreover, since 
$$S_{1/3}(L^2(\R_+),L^2_{\langle \xi \rangle^2}(\R_+)) \subset S_{1/3}(L^2(\R_+),L^2(\R_+))$$ 
holds with continuous inclusion, the Calder\'{o}n-Vaillancourt theorem gives for some $N > 0$, 
$$\|R_\lambda\|_{\mathscr{L}(\cH)} \leq C \max_{n \leq N} \|\partial_x^n \Phi'\|_\infty.$$
Using these estimates in \eqref{eq:composition} and the elliptic regularity result of Proposition \ref{prop:ellip_LhPhi} leads to
\begin{align*}
\|(\Id - \Pi_{1,\alpha})\psi\|_{\cH} &\leq C \max_{n \leq N} \|\partial_x^n \Phi'\|_\infty  \left(h^{-2/3} \|(\mathscr{L}_h^\Phi - \lambda)\psi\| + h^{2/3}\|\psi\|_{L^2(\R;\mathrm{D})} + C h^{1/3} \|K_1 u\|_{\cH}\right)\\
& \leq C \max_{n \leq N} \|\partial_x^n \Phi'\|_\infty \left(h^{-2/3} \|(\mathscr{L}_h^\Phi - \lambda)\psi\| + \|\mathscr{L}_h^\Phi\psi\|_{\cH} + h^{2/3}\|\psi\|_{\cH} + Ch^{1/3}\|K_1 u\| \right)\\
& \leq C \max_{n \leq N} \|\partial_x^n \Phi'\|_\infty \left(h^{-2/3} \|(\mathscr{L}_h^\Phi - \lambda)\psi\| + h^{2/3}\|\psi\|_{\cH} + Ch^{1/3}\|K_1 u\| \right)
\end{align*}
since $\lambda = O(h^{2/3})$. Therefore, to conclude, it remains to show that there exists $N > 0$ such that 
\begin{equation}
\label{eq:estim_crochet_poisson}
\|K_1 u\|_{\cH} \leq C \max_{n \leq N} \|\partial_x^n \Phi'\|_{\infty} \|u\|_{\cH}.
\end{equation}
To this end, we compute its symbol using Proposition \ref{prop:dxdxir}:
\begin{align*}
\{(h^{2/3}r_\lambda),p_h\} 
&= h^{2/3} \left(\partial_x r_\lambda \partial_\xi p_h - \partial_\xi r_\lambda \partial_x p_h \right) \\
&= -h^{2/3} \big(r_\lambda \partial_x p_h r_\lambda \partial_\xi p_h - r_\lambda \partial_\xi p_h r_\lambda \partial_x p_h + (r_\lambda \partial_x \pi_{1,\alpha} + \partial_x \pi_{1,\alpha})\partial_{\xi}p_h \big)\\
& =  \left(q_h [\partial_x p_h,r_\lambda] + q_h \partial_x \pi_{1,\alpha} + \partial_x \pi_{1,\alpha} q_h\right)
\end{align*}
where $q_h = (\xi + i\Phi') (h^{2/3}r_\lambda) \in S_{1/3}(L^2(\R_+),L^2(\R_+))$ (since $(h^{2/3}r_\lambda)$ ``gains'' $\langle \xi \rangle^{2}$). Moreover, 
$$[\partial_x p_h,r_\lambda] = [h^{2/3} \partial_x \mathscr{A}_\alpha(x),r_\lambda] = [\partial_x \mathscr{A}_\alpha(x),(h^{2/3}r_\lambda)] \in S_{1/3}(L^2(\R_+),L^2(\R_+)).$$
Therefore, $\{(h^{2/3}r_\lambda),p_h\} \in S_{1/3}(L^2(\R_+),L^2(\R_+))$ and it is also straightforward to check that for any $\beta \in \mathbb{N}^2$, there exists $C_\beta$ and $M_\beta$ such that 
\begin{equation}
\label{eq:semi_normes_crochet}
|\{(h^{2/3}r_\lambda),p_h\}|_{\frac13,\beta} \leq C_\beta \max_{n \leq M_\beta}\|\partial_x^n \Phi'\|_\infty
\end{equation}
using Proposition \ref{prop:rlambda_class}. The estimate \eqref{eq:estim_crochet_poisson} follows from \eqref{eq:semi_normes_crochet} and the Calder\'{o}n-Vaillancourt theorem, concluding the proof. \qed`

\section{Sharpness of Theorem \ref{thm:main}}

\label{sec:sharp}

In this section, we show that Theorem \ref{thm:main} is sharp. More precisely we show that Theorem \ref{thm:main} does {\em not} hold for $\mu = 0$.  In particular, under the same assumptions, one cannot obtain a $O(h^{s})$ localization scale for $s < \frac12$.

To show this, we produce a counter-example. Let $\alpha \equiv 1$ and let $V \in C^\infty(\R)$ satisfy Assumption \ref{ass:V} and furthermore $V_{(-\frac12,\frac12)^c} \equiv 1$.  Let 
$\psi$ the function obtained by setting $\mu = 0$ in \eqref{eq:def_phieta}, i.e.,
$$\phi = \frac{1}{\sqrt{2}} \left|\int_{0}^x \sqrt{V(s)}\,ds\right|.$$
We prove that there exists $R > 0$ and $h_0 > 0$ small enough such that for every $h \in (0,h_0)$, there exists $\lambda \in D(|z_1|e^{i\pi/3} h^{2/3},Rh)$ and $\psi \in \cH^2$ satisfying $(\mathscr{L} - \lambda) \psi = 0$ but $e^{\phi/h}\psi \notin \cH$.  

We construct $\psi$ as the following tensorized function:
$$\psi(x,y) = f(x)u_{\rm Ai}(y),$$
where $u_{\rm Ai}(y)$ is a Dirichlet eigenfunction of $D_y^2 + i y$  associated to $|z_1|e^{i \pi/3}$ (one can take $u_{\rm Ai}(y) = \textup{Ai}(ye^{i\frac\pi6}+z_1)$) and $f$ is an eigenfunction of $(hD_x)^2 + iV$ associated to a complex number $\mu \in D(0,Rh)$; such a pair $(f,\mu)$ exists for $R$ large enough and $h$ small enough, with $\mu$ furthermore satisfying $\Re(\mu) > 0$, by \cite[Proposition 3.6]{AFHR2025}. It is immediate  by construction that $\mathscr{L}_h \psi = (|z_1|e^{i \pi/3} h^{2/3} + \mu)\psi$. To show that $e^{\phi(x)/h}\psi \notin \cH$, it suffices to show that the ``amplitude'' $a(x) := e^{\frac{\phi(x)(1 +i)}h}f(x)$ does not belong to $L^2(\R)$ (by the Fubini theorem, and since $|a(x)|^2 = |e^\frac{\phi(x)}{h}f(x)|^2$). 

Indeed, suppose by contradiction that $a\in L^2(\R)$. The point is that $a$ satisfies the differential equation
\begin{equation}
	\label{eq:Cauchy}
	\left([hD_x + i(1+i)\phi']^2 + iV - \mu\right) a (x)= 0
\end{equation}
which, noticing that for $x \geq \frac12$, $V = 1$ and $\phi'(x) = \frac{1}{\sqrt{2}} $, gives
$$-h^2 a''(x) + 2e^{i\frac{\pi}{4}}ha'(x) -\lambda a(x) = 0 \quad \textup{for } x \geq \frac12.$$
Therefore, there exist $\alpha,\beta \in \mathbb{C}$ such that 
$$a(x) = \alpha e^{\theta_1 x/h} + \beta e^{\theta_2 x/h} \quad \textup{for } x \geq \frac12,$$
where $\theta_1 = e^{i \frac\pi4} + \sqrt{\mu}$ and $\theta_2 = e^{i \frac\pi4} - \sqrt{\mu}$; from this, we deduce that the condition $a \in L^2(\R)$ is only possible if $\alpha = \beta = 0$. 

In particular, $a$ satisfies the differential equation \eqref{eq:Cauchy} together with the conditions $a(1) = a'(1) = 0$; therefore, by uniqueness of the solution to this initial value problem, it follows $a$ vanishes on $\R$. This implies $f$ also vanishes on $\R$, a contradiction.

\section*{Acknowledgments}
This work was conducted within the France 2030 framework programme, Centre Henri Lebesgue ANR-11-LABX-0020-01. N.F. thanks the Région Pays de la Loire for the Connect Talent Project HiFrAn 2022 07750 led by Clotilde Fermanian Kammerer.

\appendix
\section{Proof of Lemma \ref{lem:density}}
\label{app:density}

\begin{enumerate}[1.]
\item
Let $u \in \cHd$ and let $\varepsilon > 0$. There exists a bounded subet $\Omega \subset \R \times \R_+$ such that 
$$\int_{\Omega^c} |\Delta u|^2 + |yu|^2 + |u|^2 \,dxdy \leq \varepsilon.$$
Let $\chi \in C^\infty(\overline{\R \times\R_+})$  satisfying $\chi \equiv 1$ on $\Omega$ and such that $\widetilde{\Omega} := \supp \chi$ is bounded. Then $\|u - \chi u\|_{\cHd} \leq C_\chi\varepsilon$ where $C_\chi > 0$ depends only on $\chi$.
\item Since $\widetilde{\Omega}$ is bounded, $K := \sup_{\widetilde{\Omega}} y$ is finite, and if we find $\varphi \in C^\infty(\overline{\R \times \R_+})$ satisfying $\varphi|_{\R \times \{0\}} = 0$ and $\|\varphi - u\|_{H^2(\R \times \R_+)} \leq \frac{C\varepsilon}{1+K}$ with $C$ independent of $u$, then we are done since 
$$\|u-\chi\varphi\|_{\cHd} \leq \|u - \chi u\|_{\cHd} + \|\chi (u- \varphi)\|_{\cHd} \leq C_\chi\varepsilon + (1 + K)\|u - \varphi\|_{H^2(\R^2)} \leq (C_\chi+C)\varepsilon,$$
and the condition $\varphi|_{\R \times \{0\}} = 0$ ensures that $\chi \varphi \in \cH^2$. 
\item 
By density of $C^\infty(\overline{\R} \times \R_+) \cap H^2(\R \times \R_+)$ in $H^2(\R \times \R_+)$, we can find $\varphi_1 \in C^\infty(\overline{\R} \times \R_+) \cap H^2(\R \times \R_+)$ such that
\begin{equation}\label{step3_propdensity}
\|u - \varphi_1\|_{H^2(\R\times \R_+)} \leq \frac{\varepsilon}{1 + K}.
\end{equation}
The remaining issue is that $\varphi_1$ may not vanish on $\R \times \{0\}$. However, by continuity of the trace operator $\gamma : H^2(\R \times \R_+) \to H^{3/2}(\R)$ (see \cite[Lemma 3.35]{mclean2000strongly})
\begin{equation}
\label{eq:petite_trace}
\|\gamma \varphi_1\|_{H^{3/2}(\R)} =  \|\gamma(\varphi_1 - u)\|_{H^2(\R \times \R_+)} \leq C_{\rm tr} \|\varphi_1 - u\|_{H^2(\R \times \R_+)} \leq \frac{\varepsilon C_{\rm tr}}{1 + K}
\end{equation}
where $C_{\rm tr} = \|\gamma\|_{\mathscr{L}(H^2(\R \times \R_+),H^{3/2}(\R))}$. We now exploit the smallness of the trace $\gamma \varphi_1$ to construct a small ``correction'' $\varphi_2$ such that $\varphi_1+ \varphi_2$ vanishes on $\R \times \{0\}$. 
\item Let $\varphi_2 := -\varphi_1+ w$ where $w$ is the unique solution in $H^1_0(\R\times \R_+)$ of the variational problem
$$\int_{\R \times \R_+} \nabla w \cdot \nabla v + w v = \int_{\R \times \R_+} \nabla \varphi_1 \cdot \nabla v + \varphi_1 v \quad \textup{for all } v \in H^1_0(\R \times \R_+).$$
By standard elliptic regularity of the Dirichlet Laplacian, $w \in C^\infty(\overline{\R \times \R_+}) \cap H^2(\R \times \R_+)$ and thus $\varphi_2$ also belongs to this space by linearity. Moreover, $-\Delta w + w= -\Delta \varphi_1+ \varphi_1$, thus $\Delta \varphi_2 = \varphi_2$. Hence, by \cite[Theorem 4.18]{mclean2000strongly}, there exists $C_{\rm ell} > 0$ independent of $u$ such that
\begin{equation}
\label{step4_propdensity}
\|\varphi_2\|_{H^2(\R \times \R_+)} \leq C_{\rm ell} \left( \|\varphi_2\|_{L^2(\R \times \R_+)} + \|\gamma \varphi_2\|_{H^{3/2}}\right) =  C_{\rm ell} \left( \|\varphi_2\|_{L^2(\R \times \R_+)} + \|\gamma \varphi_1\|_{H^{3/2}}\right)
\end{equation}
since $\gamma \varphi_2 = -\gamma \varphi_1$. 
\item We put $\varphi := \varphi_1 + \varphi_2$. By what precedes, $\varphi \in C^\infty(\overline{\R \times \R_+}) \cap H^2(\R \times \R_+)$ and $\varphi|_{\R \times \{0\}} = 0$. Moreover, by \eqref{step3_propdensity}, \eqref{eq:petite_trace}, \eqref{step4_propdensity} and the triangle inequality, 
\begin{equation}
\label{step5_propdensity}
\|u - \varphi\|_{H^2(\R \times \R_+)} \leq \|u - \varphi_1\|_{H^2(\R \times \R_+)} + \|\varphi_2\|_{H^2(\R_+)} \leq \frac{(1 + C)\varepsilon}{1 + K} + C_{\rm ell}\|\varphi_2\|_{L^2(\R \times \R_+)},
\end{equation}
where $C$ does not depend on $u$. Finally, using Green's theorem, \eqref{eq:petite_trace} and \eqref{step4_propdensity}
\begin{align*}
\|\varphi_2\|^2_{L^2(\R \times \R_+)}  \leq \|\varphi_2\|^2_{H^1(\R \times \R_+)} 
&= \int_{\R} \frac{\partial \varphi_2}{\partial y}(x,0) \varphi_1(x,0)\,dx \\
&\leq \|\gamma (\partial_y \varphi_2)\|_{L^2(\R)} \|\gamma \varphi_1\|_{L^2(\R)} \leq \frac{C\varepsilon}{1+K}(\|\varphi_2\|_{L^2} + \frac{C\varepsilon}{1+K}),
\end{align*}
where $C$ is independent of $u$. This implies that $\|\varphi_2\|_{L^2(\R \times \R_+)} \leq \frac{C\varepsilon}{1+K}$ where $C$ is independent of $u$. Inserting this estimate in \eqref{step5_propdensity} concludes the proof. \qedhere
\end{enumerate}

\section{Elliptic regularity estimates}

\label{app:er}

\begin{proposition}[Basic elliptic estimate for $\mathscr{L}_h^\Phi$]
\label{prop:basic_ellip}
For all $\mu > 0$, there exists $C > 0$ such that the estimate
$$h^{1/3} \left(\|D_y \psi\| + \|\sqrt{y} \psi\|\right) + \|\sqrt{V} \psi\|_{\cH} + \|(hD_x) \psi\|_{\cH} \leq C\varepsilon^{-1}\|\mathscr{L}^\Phi_h \psi\|_{\cH} + C \varepsilon \|\psi\|$$
holds for all $h > 0$, $\varepsilon \in (0,1)$, $\psi \in \cHd$, and $\Phi$ any $\mu$-subsolution in the sense of Definition~\ref{def:etaSubsolution}. 
\end{proposition}
\begin{proof}
If $\Phi$ is a $\mu$-subsolution, then the proof of \cite[Proposition 2.2]{AFHR2025} shows that
$$\Re \langle e^{-i\frac{\pi}{4}}\mathscr{L}_h^\Phi \psi,\psi\rangle \geq h^{2/3} (\|D_y \psi\|^2 + \|\sqrt{y} \psi\|^2) + \frac{\mu}{2}\|\sqrt{V} \psi\|^2.$$
In particular, by Young's inequality, for any $\varepsilon \in (0,1)$, 
\begin{equation}
\label{eq:controlV1/2}
h^{1/3} (\|D_y \psi\| + \|\sqrt{y} \psi\|) + \|\sqrt{V} \psi\|_{\cH} \leq C\varepsilon^{-1}\|\mathscr{L}_h^\Phi \psi\|_{\cH} + C\varepsilon\|\psi\|_{\cH}.
\end{equation}
In turn, since 
$$\Re \langle \mathscr{L}_h^\Phi \psi,\psi\rangle \geq \|(hD_x) \psi\|^2 - \|\Phi' \psi\|^2$$
and since $\Phi'^2 \leq \frac{1 - \mu}{2}V$ by assumption, 
\begin{equation}
\label{eq:controlxi}
\|(hD_x) \psi\| \leq \frac{1-\mu}{2}\|\sqrt{V}\psi\|^2 + \Re \langle \mathscr{L}_h^\Phi \psi,\psi\rangle \leq C \varepsilon^{-1} \|\mathscr{L}_h^\Phi \psi\|_{\cH} + C\varepsilon\|\psi\|_{\cH}
\end{equation}
by \eqref{eq:controlV1/2} and Young's inequality. The proof is concluded by summing \eqref{eq:controlV1/2} and \eqref{eq:controlxi}. 
\end{proof}

\begin{proposition}[Elliptic regularity for the Airy operator]
\label{prop:ellip_A_yD2y}
There exists $C > 0$ such that for all $u \in \mathrm{D}$, 
$$\|yu\|_{L^2(\R_+)} + \|D_y^2 u\|_{L^2(\R_+)}\leq C\|u\|_{\mathrm{D}}.$$
\end{proposition}
\begin{proof}
It is sufficient to show that 
\begin{equation}
\label{eq:ellip_reg_D2y}
\|D_y^2 u\|_{L^2(\R_+)}\leq \|u\|_{\mathrm{D}}
\end{equation}
Indeed, recalling the definition of the $\mathrm{D}$ norm from \eqref{eq:def_normeD}, one can then use the triangle inequality to obtain the other estimate 
$$\|yu\| \leq \|D^2_y u\|_{L^2(\R_+)} + \|\mathscr{A}u\|_{L^2(\R_+)} \leq C\|u\|_{\mathrm{D}}.$$

To show \eqref{eq:ellip_reg_D2y}, we use the classical method of difference quotients of Nirenberg \cite{nirenberg1955remarks}. Namely, for $u \in L^2(\R_+)$ and $h \in \R$, denote $\Delta_h u := \frac{u(y+h) - u(y)}{ih}$. 
Recall that if $u \in H^1(\R_+)$, then $\Delta_h u$ converges to $D_y u$ in $L^2$ as $h \to 0$. Given $u \in \mathrm{D}$, we write
\begin{align}\nonumber
(\mathscr{A}u,D_y^2 u) = \lim_{h \to 0} (\mathscr{A}u,\Delta_{-h} \Delta_{h} u) &= \lim_{h \to 0} (\Delta_{h} \mathscr{A}u, \Delta_h u)\\
& = \lim_{h \to 0}(\mathscr{A} \Delta_{h} u,\Delta_h u) + ([\mathscr{A},\Delta_{h}]u,\Delta_h u).
\label{eq:lim0_ellipA}
\end{align}
But on the one hand, 
$[\mathscr{A},\Delta_{h}]u = [iy,\Delta_{h}] u = iu(\cdot - h),$
so that 
\begin{equation}
\label{eq:lim1_ellipA}
\lim_{h \to 0} ([\mathscr{A},\Delta_{h}]u,\Delta_hu) = i(u,D_y u)_{L^2(\R_+)}
\end{equation}
On the other hand 
\begin{equation}
\label{eq:lim2_ellipA}
\Re (\mathscr{A} \Delta_{h}u,\Delta_h u)= \|D_y (\Delta_hu)\|_{L^2(\R_+)}^2 = \|\Delta_h (D_y u)\|^2_{L^2(\R_+)} \to_{h \to 0} \|D_y^2 u\|^2_{L^2(\R_+)}.
\end{equation}
Thus, by \eqref{eq:lim0_ellipA}-\eqref{eq:lim2_ellipA}, the Cauchy-Schwarz and the Young inequalities,  
$$\|D_y^2u\|^2 \leq C \left(\|u\|^2_{L^2(\R_+)} + \|D_y u\|^2_{L^2(\R_+)} + \|\mathscr{A}u\|^2_{L^2(\R_+)}\right)$$
and the conclusion follows since $\|D_y u\|^2_{L^2(\R_+)} = \textup{Re}(\mathscr{A}u,u) \leq \frac12( \|\mathscr{A}u\|_{L^2(\R_+)}^2 + \|u\|^2_{L^2(\R_+)}).$ 
\end{proof}
\begin{proposition}[Elliptic regularity for $\mathscr{L}_h^\Phi$]
\label{prop:ellip_LhPhi}
There exists $C > 0$ such that for all $\psi \in\cHd$ and all $h > 0$, 
\begin{equation}
\label{eq:ellip_reg_LhPhi}\left(\int_{\R}\|u(x,\cdot)\|^2_{\mathrm{D}}\,dx \right)^{1/2}\leq C\big(\|\psi\|_{\cH} + h^{-\frac23}\|\mathscr{L}_h^\Phi \psi\|_{\cH}\big).
\end{equation}
\end{proposition}
\begin{proof}
We start by showing that 
\begin{equation}
\label{eq:control_y}
\|\varphi(y)^2 \psi\|_{\cH} \leq C \big(\|\psi\|_{\cH} + h^{-\frac23}\|\mathscr{L}_h^\Phi \psi\|_{\cH}\big)
\end{equation}
where $\varphi : \R\to \R_+$ is smooth and satisfies $\varphi(y) = \sqrt{y}$ for $y \in [1,+\infty)$; this immediately implies the same estimate with $\varphi(y)^2$ replaced by $y$. 

To this end, let $\chi \in C^\infty_c(\R)$ be real-valued with $\chi \equiv 1$ near $0$ and let $\varphi_\delta(y) := \chi(\delta y)\varphi(y)$ and  By the dominated convergence theorem, 
$$(\mathscr{L}_h^\Phi \psi, \varphi(y)^2\psi) = \lim_{\delta \to 0} (\mathscr{L}_h^\Phi \psi, \varphi_\delta^2 \psi) =\lim_{\delta \to 0} \big( \mathscr{L}_h^\Phi(\varphi_\delta(y)\psi),\varphi_\delta(y) \psi\big) + \big([\mathscr{L}_h^\Phi,\varphi_\delta(y)]\psi,\varphi_\delta(y)\psi\big)$$ 
We observe that
$$\Im (\mathscr{L}_h^\Phi \varphi_\delta(y)\psi,\varphi_\delta\psi) = h^{2/3}(\alpha(x)\varphi_\delta(y)\psi,\varphi_\delta(y)\psi) \geq ch^{2/3}\|\varphi_\delta(y) \psi\|^2_{\cH}$$
since $\alpha$ is bounded below. Moreover, since $\varphi'$ is bounded, $\|[\mathscr{L}_h^\Phi ,\varphi_\delta] \psi\| \leq Ch^{2/3} (\|\psi\|_{\cH} + \|D_y \psi\|_{\cH})$. Thus, 
$$\lim_{\delta \to 0} h^{2/3} \|\varphi_\delta(y) \psi\|^2_{\cH} \leq \|\mathscr{L}_h^\Phi \psi\|\|\varphi_\delta(y)^2 \psi\| +Ch^{2/3}(\|\psi\|_{\cH} + \|D_y\psi\|_{\cH})\|\varphi(y) \psi\|_{\cH}.$$
Using the basic elliptic estimate of Proposition \ref{prop:basic_ellip} with $\varepsilon = h^{1/3}$ and using Young's inequality, we deduce that
$$h^{2/3}\|\varphi_\delta(y) \psi\|^2 \leq C h^{-2/3}\|\mathscr{L}_h^\Phi \psi\|_{\cH}^2 + h^{2/3}\|\psi\|^2_{\cH}\,.$$
Sending $\delta$ to $0$, we obtain \eqref{eq:control_y}. 

We obtain the estimate
\begin{equation}
\label{eq:controlD2y}
\|D_y^2 \psi\|_{\cH} \leq C \left(\|\psi\|_{\cH} + h^{-\frac23} \|\mathscr{L}_h^\Phi \psi\|_{\cH}\right)
\end{equation}
by the method of difference quotients as in the proof of Proposition \ref{prop:ellip_A_yD2y}. We omit the detail, since no new difficulty arises. Summing \eqref{eq:control_y} and \eqref{eq:controlD2y} and using the Fubini theorem 
$$\int_{\R} \|y \psi(x,\cdot)\|^2_{L^2(\R_+)} + \|D_y^2 \psi(x,\cdot)\|^2_{L^2(\R_+)} \,dx \leq C \big(\|\psi\|_{\cH} + h^{-2/3} \|\mathscr{L}_h^\Phi \psi\|_{\cH}\big)$$
which implies \eqref{eq:ellip_reg_LhPhi} by the definition of the $\mathrm{D}$ norm and the triangle inequality. 
\end{proof}

\section{Existence of eigenvalues of \texorpdfstring{$\mathscr{L}_h$}{Lh} close to \texorpdfstring{$\lambda_{1,\alpha}(0)h^{2/3}$}{lambda\_1,alpha}}

\label{app:sp_non_vide}

		In this paragraph, we prove the following result
		\begin{proposition}
		\label{prop:D}
		There exists $R > 0$ and $h_0 > 0$ such that for all $h \in (0,h_0)$, $\mathscr{L}_h$ admits at least one eigenvalue in the dist $D(\lambda_{1,\alpha}(0)h^{\frac23}, Rh)$. 
		\end{proposition}
		\begin{proof}
		We use the following three steps, which are shown individually in separate lemmas below. 
				\begin{enumerate}[(i)]
				\item We show that the spectrum of $\mathscr{L}_h$ in a small disk $D(0,c)$ consists exclusively of eigenvalues (Lemma \ref{lem:spectre_vp_D(O,c)}).
				\item We show that if $\lambda \in D(\lambda_{1,\alpha}(0)h^{2/3},Rh)$ satisfies 
				$$\dist(\lambda - \lambda_{1,\alpha}(0)h^{2/3},S_h) \gtrsim h$$
				where $S_h = \big\{e^{i \frac{\pi}{4}}(2n-1)\kappa h \,\mid\, n = 1,2,\ldots\big\}$ is the spectrum of the complex harmonic oscillator $h^2D_x^2 + i \kappa^2x^2$, with $\kappa = \sqrt{\frac{V''(0)}{2}}$,  then $\lambda \in \rho(\mathscr{L}_h)$ and 
				$$\|(\mathscr{L}_h - \lambda)^{-1}\| \lesssim h^{-1},$$
				see Lemma \ref{lem:oscillateur}.
				\item By (ii), there exists a small circle $C := \mathscr{C}(\lambda_{1,\alpha}(0)h^{\frac23} + e^{i \frac{\pi}{4}}\kappa h,\varepsilon h)$ in the resolvent set of $\mathscr{L}_h$. Thus, the operator
				$$P_h := \frac{1}{2\pi i}\int_{C}(z - \mathscr{L}_h)^{-1}\,dz$$
				is well-defined. We then exhibit an explicit tensorized ``quasimode'' $\psi$ for which we prove that $P_h \psi \neq 0$ (Lemma \ref{lem:quasimode}). Thus $P_h \neq 0$, which means that $(z - \mathscr{L}_h)^{-1}$ cannot be holomorphic in a closed disk containing $C$. In particular, the spectrum of $\mathscr{L}_h$ cannot be empty in the disk $D(\lambda_{1,\alpha}(0)h^{2/3},Rh)$ and by (i), it follows that $\mathscr{L}_h$ admits an eigenvalue in this disk.
				\end{enumerate}
		\end{proof}
		
		\begin{lemma}
		\label{lem:spectre_vp_D(O,c)}
		There exists $c > 0$ and $h_0 > 0$ such that for all $h \in (0,h_0)$, the spectrum of $\mathscr{L}_h$ in $D(0,c)$ is discrete and consists only of eigenvalues of $\mathscr{L}_h$. Moreover, for all $\lambda \in D(0,c)$, $\mathscr{L}_h - \lambda$ is Fredholm of index $0$.
		\end{lemma}
		\begin{proof}
		It can be shown by a standard argument that $\mathscr{L}_h - z$ is Fredholm of index $0$ for $z \in D(0,c)$ with $c$ small enough (by adding a smooth, compactly supported perturbation $\chi$ equal to a positive constant on the set where $V(x) + y \leq \frac{\liminf_{\infty} V}{2}$). For $z = -\frac{c}{2}$, we have $\Re \langle (\mathscr{L}_h + c/2) \psi,\psi \rangle \geq \frac{c}{2}\|\psi\|^2$, thus $\mathscr{L}_h + \frac{c}{2}$ is injective, hence an isomorphism; in particular $-\frac{c}{2} \in \rho(\mathscr{L}_h)$. Since $z \mapsto \mathscr{L}_h - z$ is holomorphic with respect to $z$, the Fredholm analytic theorem (see, e.g., \cite[Theorem C.9]{dyatlov2019mathematical}) implies that the spectrum of $\mathscr{L}_h$ is discrete and consists only of elements with finite algebraic multiplicity, hence eigenvalues (by \cite[Corollary 3.36]{cheverry2021guide}).
		\end{proof}

		\begin{lemma}
		\label{lem:oscillateur}
		If $\lambda \in D(\lambda_{1,\alpha}(0)h^{2/3},Rh)$ satisfies 
				$$\dist(\lambda - \lambda_{1,\alpha}(0)h^{2/3},S_h) \gtrsim h$$
				where $S_h = \big\{e^{i \frac{\pi}{4}}(2n-1)\kappa h \,\mid\, n = 1,2,\ldots\big\}$ is the spectrum of the complex harmonic oscillator $h^2D_x^2 + i \kappa^2x^2$, with $\kappa = \frac{V''(0)}{2}$,  then $\lambda \in \rho(\mathscr{L}_h)$ and 
				$$\|(\mathscr{L}_h - \lambda)^{-1}\| \lesssim h^{-1}$$
		\end{lemma}
		\begin{proof}
		We can use a particular case of Proposition \ref{prop:ellip2} (applied with $\Phi=0$) to get
		\begin{equation}
		\label{eq:aux1}
		h^{\frac23}\|(\mathrm{Id}-\Pi_{1,\alpha})\psi\|\leq C\|(\mathscr{L}_h-\lambda)\psi\|+Ch\|\psi\|\,.
		\end{equation}
		Moreover, using Corollary \ref{cor:commut} (again with $\Phi = 0$)
		\begin{equation}
		\label{eq:aux2}
		\|(\mathscr{L}_h - \lambda)\Pi_{1,\alpha} \psi\| \leq C\|(\mathscr{L}_h - \lambda)\psi\|+ Ch^{\frac43} \|\psi\|.
		\end{equation}
		We again observe (as in the proof of Proposition \ref{prop:ellip1}) that
		$$(\mathscr{L}_h - \lambda)\Pi_{1,\alpha} = [(hD_x)^2 + V_{h,{\rm eff}} - z]\Pi_{1,\alpha}$$ 
		where $V_{h,{\rm eff}} =  iV + h^{2/3} (\lambda_{1,\alpha}(x) - \lambda_{1,\alpha}(0))$ and, here, $z \in D(0,Rh)$ satisfies $\dist(z,S_h) \geq \varepsilon h$. Thus, by \cite[Corollary 3.10]{AFHR2025} (adapting the proof as explained in Remark 1.2 (ii) of that reference) and using the Fubini theorem, 
		\begin{equation}
		\label{eq:aux3}
		\|(\mathscr{L}_h - \lambda)\Pi_{1,\alpha} \psi\| \geq Ch \|\Pi_{1,\alpha} \psi\|.
		\end{equation}
		Combining equations~\eqref{eq:aux1}-\eqref{eq:aux3}, we deduce that 
		$$\|(\mathscr{L}_h - \lambda) \psi\| \geq Ch \|\psi\|.$$
		Thus $(\mathscr{L}_h -\lambda)$ is injective, hence an isomorphism by Lemma \ref{lem:spectre_vp_D(O,c)} and the conclusion follows.		
		\end{proof}

		Recall the operator $P_h$ introduced in step (iii) of the proof of Proposition \ref{prop:D}. 
		\begin{lemma}
		\label{lem:quasimode}
		The operator $P_h$ is not equal to $0$.  
		\end{lemma}
		\begin{proof}
		Let $\psi = u_{{\rm Airy}}(y)f_h(x)$ where 
		$$u_{\rm Airy}(y) :=  \textup{Ai}(e^{i \frac\pi6} \alpha(0)^{1/3}y + z_1)\,, \quad f_h(x) = h^{-1/4}e^{-cx^2/h}$$
		where $c = e^{i\pi/4}\kappa$. The point is that
		$$(D_y^2 + i \alpha(0)y)u_{\rm Airy} = \lambda_{1,\alpha}(0) u_{\rm Airy} \quad \textup{and} \quad [(hD_x)^2 + i\kappa^2 x^2]f_h = e^{i \frac\pi4}\kappa h f_h\,,$$
		Thus, using Taylor expansions of $V$ and $\alpha$ near $x = 0$, one can check that 
		$$(\mathscr{L}_h -\mu(h))\psi = O(h^{3/2})\|\psi\|,$$
		where $\mu_1(h) :=  \lambda_{1,\alpha}(0)h^{2/3} + e^{i \frac\pi4}\kappa h$. 
		Therefore, 
		\[\begin{split}\left\|(P_h-\mathrm{Id})\psi\right\|&=\left\|\frac{1}{2i\pi}\int_{C}\left((\zeta-\mathscr{L}_h)^{-1}-(\zeta-\mu_1(h))^{-1}\right)\psi\,{d}\zeta\right\|\\
		&=\left\|\frac{1}{2i\pi}\int_{C}\left((\zeta-\mathscr{L}_h)^{-1}(\zeta-\mu_1(h))^{-1}\right)\,{d}\zeta\,(\mathscr{L}_h-\mu_1(h))\psi\,\right\|\\
		& \leq Ch^{-1} h^{3/2}|\psi\| = Ch^{1/2} \|\psi\|.
		\end{split}\]
		In particular, $P_h \neq 0$, concluding the proof.
		\end{proof}

\section{Pollution by the Airy operator}

\label{app:pollution}

In this paragraph, we outlie the pollution phenomenon alluded to in Remark \ref{rem:comments_main}(\ref{remitem:pollution}). Firstly, it is possible to obtain a weaker localization result than Theorem \ref{thm:main}, namely, a localization at scale $O(h^{1/3})$, by elementary manipulations of the quadratic form. Indeed it is not difficult to show that for $\lambda \in D(\lambda_{1,\alpha}(0)h^{\frac23},Rh)$,
$$\textup{Re} \left[e^{-i\frac\pi4} \left\langle  \left(e^{\Phi(x)/{h^{2/3}}}(\mathscr{L}_h - \lambda)e^{-\Phi(x)/h^{2/3}}\right)\psi,\psi\right\rangle\right] \geq c  \left \langle (V -Ch^{2/3})  \psi,\psi \right\rangle$$
where $c > 0$ and $C > 0$ are independent of $h$, and $\Phi$ is a smooth bounded function satisfying $\Phi(x) \geq \min(1,x^2)$. Here, the term $Ch^{2/3}$ is simply due to the fact that $\lambda = O(h^{2/3})$. This leads to the Agmon estimate
$$\|e^{\Phi/h^{2/3}} \psi\| \leq C \|\psi\|$$
for any eigenfunction $\psi$ associated to $\lambda$. Since $V(x) \sim x^2$ as $x \to 0$, this gives an $O(h^{1/3})$ localization in the $x$ variable.

To understand why this argument fails to obtain the optimal $O(h^{1/2})$ localization scale, it is instructive to consider an analogous situation involving a self-adjoint counterpart of $\mathscr{L}_h$ -- namely, when we replace $\mathscr{L}_h$ by the operator 
$$\mathscr{L}_h^{\rm sa} := h^{2/3}(D^2_y + \alpha(x) y) + (hD_x)^2 + V(x)$$
In this case, the variational argument above can be improved by exploiting the fact that, by the min-max principle, one has (in the sense of quadratic forms)
$$D_y^2 + \alpha(x)y \geq \mu_\alpha(x)$$ 
where $\mu_{1,\alpha}(x) := \alpha(x)^{2/3}|z_1|$. Therefore, splitting $\lambda \in D(\mu_{1,\alpha}(0)h^{2/3},Rh)$ as
$$\lambda = \mu_{1,\alpha}(x)h^{2/3} + (\mu_{1,\alpha}(0) - \mu_{1,\alpha}(x))h^{2/3} + z\,,$$
the action of $\mathscr{L}_h^{\rm sa} - \lambda$ is decoupled into an ``Airy part'' and an effective ``Shcrödinger part'' via 
$$\mathscr{L}_h^{\rm sa} - \lambda = h^{2/3}\overbrace{\left(D_y^2 + \alpha(x)y - \mu_{1,\alpha}(x)\right)}^{\textup{Airy part} \geq 0} + \overbrace{h^2 D_x^2 + V_{h,{\rm eff}}(x) - z}^{\textup{Schrödinger part}}$$
where $V_{h,{\rm eff}}(x) = V(x) + h^{2/3}\left(\mu_\alpha(0) - \mu_\alpha(x)\right)$, and, {\em crucially}, $z \in D(0,Rh)$. Thus, we are in the situation of the proof of Proposition \ref{prop:ellip1}, but the difference is that this decoupling is not restricted to the image of the adiabatic projection $\Pi_{1,\alpha}$. 

In contrast, in the complex setting, the ``complex Airy part'' $A := D_y^2 + i \alpha(x)y - \lambda_{1,\alpha}(x)$ {\em does not have a sign}. Worse still, there exists $\varepsilon > 0$ such that
\begin{equation}
\label{eq:claim_numrange}
\left\{\frac{ \langle A\psi,\psi\rangle}{\|\psi\|^2_{\cH}}\,\,\Big|\,\, \psi \in \cH^2 \right\} \supset D(0,\varepsilon)
\end{equation}
i.e., the numerical range of $A$, contains a neighborhood of $0$ (this is an immediate consequence of Corollary \ref{cor:num_range} below). This leaves little hope of improving the $O(h^{1/3})$ localization result outlined in point (i) above by limiting oneself to direct arguments on the quadratic form.

We now show the claim \eqref{eq:claim_numrange}. As in \S\ref{sec:Pialpha}, we denote by $\mathscr{A} := D_y^2 + iy$ the complex Airy operator on the Hilbert space $L^2(\R_+)$ with domain $\mathrm{D}$ defined by \eqref{eq:def_D}. The {\em numerical range} of $\mathscr{A}$ is the subset of the complex plane defined by
$$W(\mathscr{A}) := \left\{\langle \mathscr{A}\psi,\psi\rangle\,\mid\, \psi \in \mathrm{D}\, \textup{ and } \|\psi\|_{L^2(\R_+)}=1\right\}.$$
The smallest eigenvalue in magnitude of $\mathscr{A}$ is given by $|z_1| e^{i \frac\pi3}$; however, this value is not on the boundary of $W(\mathscr{A})$. To show this, we start by the following lemma, which is reminiscent of the virial theorem. 
\begin{proposition}
\label{lem:virial}
Let $u_1$ be an $L^2$-normalized eigenfunction of the {\em self-adjoint} Airy operator $A = D_y^2 + y$ for its smallest eigenvalue $|z_1|$. Then
$$\|D_y u_1\|^2_{L^2(\R_+)} = \frac12 \langle yu_1,u_1\rangle_{L^2(\R_+)} = \frac{|z_1|}{3}.$$
\end{proposition}
\begin{proof}
Let $a := \|D_y u_1\|^2_{L^2(\R_+)}$ and $b := \langle yu_1,u_1\rangle_{L^2(\R_+)}$. We have $a + b = |z_1|$. Moreover, for all $\gamma >  0$, consider $u_\gamma := u_1(\gamma y)$. Then by homogeneity,
$$f(\gamma) := \frac{\langle A u_\gamma,u_\gamma\rangle}{\|u_\gamma\|_{L^2(\R_+)}^2} = \gamma^{-2} a + \gamma b.$$
The min-max principle states that $f$ attains a global minimum for $\gamma = 1$. Thus $f'(1) = 0$, that is, $b = 2a$. 
\end{proof}

\begin{corollary}
\label{cor:num_range}
There exists $\varepsilon > 0$ such that $D(|z_1| e^{i\frac\pi3},\varepsilon) \subset W(\mathscr{A})$.
\end{corollary}
\begin{proof}
By Lemma \ref{lem:virial} and by homogeneity, for all $\gamma > 0$, $W(\mathscr{A})$ contains the complex number
$$z(\gamma) := \frac{\langle \mathscr{A}u_\gamma,u_\gamma\rangle}{\|u_\gamma\|_{L^2(\R_+)^2}} = \frac{|z_1|}{3\gamma^2} + 2i \frac{\gamma |z_1|}{3}\,,$$
where $u_\gamma(x) := u_1(\gamma x)$. Observing that
$$\frac{|z_1|}{3} < \Re(|z_1|e^{i\frac\pi3}) \quad \textup{and}\quad  \frac{2|z_1|}{3} < \Im(|z_1|e^{i\frac\pi3}),$$
one can check that the triangle with vertices $z(\gamma)$, $z(\gamma^{-1})$ and $z(1)$ `contains $|z_1| e^{i\pi/3}$ in its interior for $\gamma$ large enough. Since $W(\mathscr{A})$ is convex (by the Toeplitz-Haussdorff theorem \cite{toeplitz1918algebraische,hausdorff1919wertvorrat}), the conclusion follows. 
\end{proof}

\bibliographystyle{abbrv}
\bibliography{biblio.bib}

\end{document}